%% file: landscape_main.tex
\documentclass[a4paper,UKenglish]{lipics}

\usepackage{wrapfig}
\usepackage{fontenc}
\usepackage{array}

\usepackage{multirow} 
\usepackage{cite}
\usepackage{amsfonts}
\usepackage{amstext}
\usepackage{mathtools}
\usepackage{paralist}

\usepackage{tikz}
\usetikzlibrary{calc}
\usetikzlibrary{shapes,arrows,positioning,shadows,snakes}

\usepackage{footmisc}

\usepackage{amsthm} 

\usepackage{pifont}

 \providecommand{\tabularnewline}{\\}

\usepackage{nameref}
\hypersetup{linktocpage=true,pagebackref=true}
\usepackage{cleveref}

\usepackage{thmtools,thm-restate} 

\usepackage{graphicx}
\usepackage{graphics}
\usepackage{microtype}
\usepackage{colordvi}
\usepackage{xspace}
\usepackage{algorithm}
\usepackage{algorithmicx}
\usepackage{algpseudocode}
\usepackage{url}
\usepackage{enumitem}
\usepackage{authblk}

\newcommand{\WDF}{{\mathit{WDF}}}
\newcommand{\AL}{\mathit{WB}}
\newcommand{\WB}{\mathit{WB}}
\newcommand{\SO}{\mathit{SO}}
\newcommand{\SF}{\mathit{SF}}
\newcommand{\WSF}{\mathit{WSF}}
\newcommand{\FF}{\mathit{FF}}
\newcommand{\LF}{{\mathit LF}}
\newcommand{\DF}{\mathit{DF}}
\newcommand{\WS}{\mathit{WS}}

\newcommand{\UB}{\mathit{UB}}

        {\hspace*{\fill}$\Box$\par\vspace{4mm}}

\makeatletter
\def\thmt@refnamewithcomma #1#2#3,#4,#5\@nil{%
  \@xa\def\csname\thmt@envname #1utorefname\endcsname{#3}%
  \ifcsname #2refname\endcsname
    \csname #2refname\expandafter\endcsname\expandafter{\thmt@envname}{#3}{#4}%
  \fi
}
\makeatother

\newcommand{\ceil}[1]{\ensuremath{\left\lceil#1\right\rceil}}

\newcommand{\abs}[1]{\lvert #1\rvert}

\newcommand{\OPT}{\mbox{\sf OPT}}

\newcommand{\opt}{\mbox{\sf OPT}}

\newcommand{\set}[1]{\left\{ #1 \right\}}

\newcommand{\tset}{{\mathcal T}}

\newcommand{\pset}{{\mathcal{P}}}

\newcommand{\bset}{{\mathcal{B}}}
\newcommand{\aset}{{\mathcal{A}}}

\newcommand{\sset}{{\mathcal{S}}}

\newcommand{\be}{\begin{enumerate}}
\newcommand{\ee}{\end{enumerate}}
\newcommand{\bd}{\begin{description}}
\newcommand{\ed}{\end{description}}
\newcommand{\bi}{\begin{itemize}}
\newcommand{\ei}{\end{itemize}}

\renewcommand{\phi}{\varphi}
\newcommand{\eps}{\epsilon}

\newcommand{\N}{\ensuremath{\mathbb N}}

\newcommand{\Greedy}{\textsc{Greedy}}

\newcommand{\kiopt}{\text{KI-OPT}}
\newcommand{\kimtr}{\text{KI-MTR}}

\def\markatright#1{\leavevmode\unskip\nobreak\quad\hspace*{\fill}{#1}}
\renewenvironment{proof}
  {\begin{trivlist}\item[\hskip\labelsep{\emph{Proof}.}]}
  {\markatright{\qed}\end{trivlist}}

\theoremstyle{remark}

\newtheorem{fact}{Fact}

\def\ShowComment{True}

\ifdefined\ShowComment

\def\parinya#1{\marginpar{$\leftarrow$\fbox{P}}\footnote{$\Rightarrow$~{\sf #1 --Parinya}}}

\else

\def\parinya#1{}

\fi

\ifdefined\ShowComment

\def\kurt#1{\marginpar{$\leftarrow$\fbox{K}}\footnote{$\Rightarrow$~{\sf #1 --Kurt}}}

\else

\def\kurt#1{}

\fi

\ifdefined\ShowComment

\def\thatchaphol#1{\marginpar{$\leftarrow$\fbox{T}}\footnote{$\Rightarrow$~{\sf #1 --Thatchaphol}}}

\else

\def\thatchaphol#1{}

\fi

\ifdefined\ShowComment

\def\laszlo#1{\marginpar{$\leftarrow$\fbox{L}}\footnote{$\Rightarrow$~{\sf #1 --LK}}}

\else

\def\laszlo#1{}

\fi

\ifdefined\ShowComment

\else

\def\laszlo#1{}

\fi

\newcommand{\Kurt}[1]{\textcolor{red}{#1}\marginpar{Kurt}}

\title{The landscape of bounds for binary search trees} 

\author[1]{Parinya Chalermsook} 
\author[1]{Mayank Goswami}  
\author[2]{L\'{a}szl\'{o} Kozma} 
\author[1]{Kurt Mehlhorn} 
\author[3]{Thatchaphol Saranurak}  

\affil[1]{{Max-Planck Institute for Informatics, Germany.}  
{\footnotesize \tt \{parinya,gmayank,mehlhorn\}@mpi-inf.mpg.de} }
\affil[2]{{Saarland University, Germany. } { \tt kozma@cs.uni-saarland.de}} 
\affil[3]{{KTH Royal Institute of Technology, Sweden.} {\footnotesize \tt thasar@kth.se} } 

\authorrunning{P.\ Chalermsook, M.\ Goswami, L.\ Kozma, K.\ Mehlhorn, T.\ Saranurak} 

\Copyright{Parinya Chalermsook, Mayank Goswami, L\'{a}szl\'{o} Kozma, Kurt Mehlhorn, Thatchaphol Saranurak}

\serieslogo{}
\volumeinfo
  {Billy Editor and Bill Editors}
  {2}
  {Conference title on which this volume is based on}
  {1}
  {1}
  {1}
\EventShortName{}
\DOI{10.4230/LIPIcs.xxx.yyy.p}

\begin{document}

\maketitle

\input{abstract2}

\input{intro}

\input{dictionary}

\input{new-eq}

\input{lazy-finger}

\input{interleaving}

\input{examples}

\newpage 
\bibliographystyle{plain}
\bibliography{ref}

\newpage 
\appendix 
\clearpage
\input{appendix}

\input{lazy-finger_appendix}
\input{interleaving_appendix}

\input{examples_appendix}

\end{document}

%% file: abstract2.tex
\begin{abstract}

Binary search trees (BSTs) with rotations can adapt to various kinds of structure in search sequences, achieving amortized access times substantially better than the $\Theta(\log{n})$ worst-case guarantee.
Classical examples of structural properties include \emph{static optimality}, \emph{sequential access}, \emph{working set}, \emph{key-independent optimality}, and \emph{dynamic finger}, all of which are now known to be achieved by the two famous online BST algorithms (Splay and Greedy). 
Beyond the insight on how ``efficient sequences'' might look like, structural properties are important as stepping stones towards proving or disproving dynamic optimality, the elusive 1983 conjecture of Sleator and Tarjan that postulates the existence of an asymptotically optimal online BST. A BST can be optimal only if it satisfies all ``sound'' properties (those achieved by the \emph{offline} optimum).

In this paper, we introduce novel properties that explain the efficiency of sequences not captured by any of the previously known properties, and which provide new barriers to dynamic optimality. We also establish connections between various properties, old and new. For instance, we show the following.
\begin{compactitem}
\item A tight bound of $O(n \log{d})$ on the cost of Greedy for $d$-decomposable sequences, improving our earlier $n 2^{O(d^2)}$ bound (FOCS 2015). 
The result builds on the recent lazy finger result of Iacono and Langerman (SODA 2016). 
On the other hand, we show that lazy finger alone cannot explain the efficiency of pattern avoiding sequences even in some of the simplest cases.

\item  
A hierarchy of bounds using multiple lazy fingers, addressing a recent question of Iacono and Langerman.

\item The optimality of the Move-to-root heuristic in the key-independent setting introduced by Iacono (Algorithmica 2005).

\item A new tool that allows combining any finite number of sound structural properties. 
As an application, we show an upper bound on the cost of a class of sequences that all known properties fail to capture.  

\item The equivalence between two families of BST properties. The observation on which this connection is based was known before -- we make it explicit, and apply it to classical BST properties.     
This leads to a clearer picture of the relations between BST properties and to a new proof of several known properties of Splay and Greedy that is arguably more intuitive than the current textbook proofs.
\end{compactitem}
\end{abstract}

%% file: intro.tex
\section{Introduction}

In the dynamic BST model a sequence of keys are accessed in a binary search tree, and after each access, the tree can be reconfigured via a sequence of rotations and pointer moves starting from the root. (There exist several alternative but essentially equivalent cost models, see e.g.~\cite{Wilber,DHIKP09}.)
Two classical online algorithms in this model are the Splay tree of Sleator and Tarjan~\cite{ST85} and Greedy, an algorithm discovered independently by Lucas~\cite{Luc88} and Munro~\cite{Mun00} and turned into an online algorithm by Demaine et al.~\cite{DHIKP09}.

Our understanding of the BST model goes far beyond the usual paradigm of worst-case complexity. For broad classes of access sequences the $\Theta(\log{n})$ worst-case bound is too pessimistic, and both Splay and Greedy are able to achieve better amortized access times. 
Understanding the kinds of structure in sequences that facilitate efficient access has been the main focus of BST research in the past decades.  
The description of useful structure is typically given in the form of \emph{formulaic bounds}.  

Given an access sequence $X = (x_1,\ldots, x_m) \in [n]^m$, a {\em formulaic BST bound} (or simply BST bound) is a function $f(X)$ computable in polynomial time, intended to capture the access cost of BST algorithms on sequence $X$. We say that a BST bound $f(X)$ is {\em sound} if $\opt(X) \leq O(f(X)+|X|)$ where $\opt(X)$ is the optimal cost achievable by an \emph{offline} algorithm. The bound $f(X)$ is {\em achieved by algorithm $\aset$} if the cost of accessing sequence $X$ by $\aset$ (denoted $\aset(X)$) is at most $O(f(X)+|X|)$. 
BST bounds play two crucial roles: 

(i) They shed light on the structures that make sequences efficiently accessible by BST algorithms. 
For instance, the \emph{dynamic finger bound} $\sum_{i} \log |x_i - x_{i+1} + 1|$ intuitively captures the ``encoding length'' of the distances between consecutive accesses (algorithms can take advantage of the proximity of keys).

(ii) A sound BST bound is a concrete intermediate step towards the {\em dynamic optimality conjecture}~\cite{ST85}, which postulates that a simple online algorithm can asymptotically match the optimum on every access sequence, i.e.\ that it can be $O(1)$-competitive. This has been conjectured for both Splay and Greedy, but the conjecture remains unsettled after decades of research.
An $O(1)$-competitive algorithm needs to achieve all sound BST bounds. 
Proposing concrete bounds and verifying whether candidate algorithms such as Splay or Greedy achieve them has been so far the main source of progress towards dynamic optimality. 
 
Several such bounds appear in the literature: besides the classical  
{\em dynamic finger}~\cite{ST85, finger2}, {\em working set}~\cite{ST85,Iacono05}, {\em unified bound}~\cite{ST85,ElmasryFI13}, etc.\ recently studied bounds include {\em lazy finger} and {\em weighted dynamic finger}~\cite{BoseDIL14,LI16}, and bounds pertaining to \emph{pattern-avoidance}~\cite{ChalermsookG0MS15}. In some cases the interrelation between these bounds is unclear (i.e.\ whether one subsumes the other). 

\smallskip\noindent{\bf Our contributions.} 
In this paper we systematically organize the known bounds into a coherent picture, and study the pairwise relations between bounds. We introduce new sound BST bounds (in fact, a hierarchy of them). Some of these bounds serve as bridges between existing bounds, whereas others explain the easiness of certain sequences, hitherto not captured by any known bound. (We only focus on bounds defined on access sequences, we ignore therefore the deque~\cite{tarjan_sequential}, and split~\cite{split_Luc91} conjectures, that concern other operations.)

We highlight in this section the contributions that we find most interesting, with an informal discussion of their implications. We refer to \textsection\,\ref{sec:dict} for a more precise definition of the bounds considered in this paper. Our current knowledge of sound BST bounds and their relations, i.e.\ the ``landscape'' of BST bounds is presented in Figure~\ref{fig:landscape}. In the following, let $X \in [n]^m$ be an arbitrary access sequence.

\smallskip
\noindent{\bf Lazy finger results.} Our first set of contributions is a study of lazy finger bounds, their generalizations, and their connections with other BST bounds.
The lazy finger bound~\cite{BoseDIL14}, denoted $LF(X)$, captures the ``proximity'' of successive accesses \emph{in a reference tree}. Bose et al.~\cite{BoseDIL14} proved that lazy finger generalizes the classical dynamic finger bound.  
It was recently shown~\cite{LI16} that $\Greedy(X) \leq O(LF(X))$.
We prove a new connection between lazy finger and a recently studied~\cite{ChalermsookG0MS15} {\em decomposability parameter} $d(X)$.

\vspace{-0.05in}
\begin{theorem} 
For permutation sequence $X \in [n]^n$ and decomposability parameter $d = d(X)$, we have $LF(X) \leq O(n \log d)$. 
\end{theorem}  
\vspace{-0.05in}

As a corollary, we obtain the tight bound $\Greedy(X) \leq O(n \log d)$ which improves the earlier bound of $O(n 2^{O(d^2)})$ and resolves an open question from~\cite{ChalermsookG0MS15}.
We remark that $d(X)$ is a natural parameter whose 
special case $d(X) = 2$ includes the well-known 
traversal sequences.
 
Next, inspired by~\cite{DemaineILO13}, we define the $k$-lazy finger parameter $LF^k(X)$.
\vspace{-0.05in}  
\begin{theorem} 
\label{thm:mainlf}
Let $X \in [n]^m$ be a sequence and $k \in {\mathbb N}$. Then $\opt(X) \leq O(\log k) \cdot LF^k(X)$.
\end{theorem}

\vspace{-0.05in}
This improves the $O(k) \cdot LF^{k}(X)$ bound, which is implicit in~\cite{DemaineILO13}. Moreover, our bound is tight in the sense that there exists $X$ for which $\opt(X) \geq \Omega(\log k) \cdot LF^k(X)$.
Remark that $LF^1(X) \geq LF^2(X) \geq \ldots \geq LF^n(X)$, giving a hierarchy of sound BST bounds.  
For $k \geq 2$, this bound is not known to be achieved by any online algorithm.

The bounds in the $k$-lazy finger hierarchy are not implied by each other. In fact, we show a strongest possible separation between $LF^k$ and $LF^{k-1}$. That is, for any $k$, there is a sequence $X$ for which $LF^{k-1}(X)/LF^{k}(X) \geq \Omega(\log (n/k))$.  
This result yields 
a large number of intermediate steps towards dynamic optimality: A candidate algorithm not only needs to achieve the finger bounds for constantly many fingers, but also has to match the asymptotic ratio of $O(\log k)$.  

We show an application of multiple lazy fingers by giving a new upper bound on $\opt$: Let $m(X)$ denote the {\em monotone complexity} parameter of $X$.
In~\cite{ChalermsookG0MS15}, we showed that $\Greedy(X)  \leq |X| \cdot 2^{O(k^2)}$ for $k= m(X)$. 
Here we show that $\opt(X) \leq O(\log k) \cdot LF^k(X) \leq  O(k \log k) \cdot |X|$, raising the open question of whether there is any online algorithm matching this bound.  

\smallskip
\noindent{\bf Interleave results.} We introduce a simulation technique that allows combining any finite number of sound BST bounds.

\begin{wrapfigure}[13]{l}[0.15\textwidth]{0.31\textwidth}
	\begin{center}  
	\vspace{-0.05in}
		\includegraphics[width=0.25\textwidth, trim=0 0 0 0.5cm]{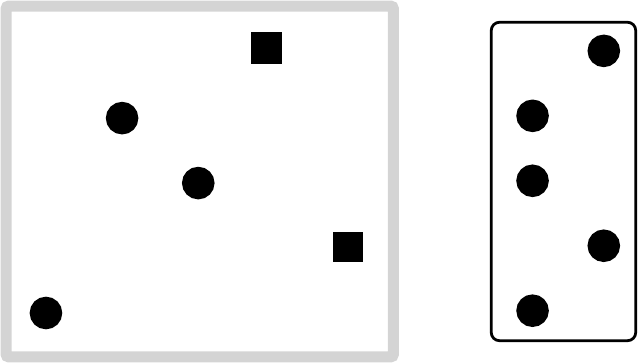}
	\end{center}
	\caption{\label{fig:interleave} The sequence $X= (1,5,3,2,4)$ (left) resulting from composing the sequences $X^{(1)} = (1,3,2)$ (circles) and $X^{(2)} = (2,1)$ (squares) with the template ${\tilde X = (1,2,1,1,2)}$ (right).}
\end{wrapfigure}

Let $X^{(1)},\ldots, X^{(\ell)}$ be a collection of sequences where $X^{(i)} \in [n_i]^{m_i}$ (each $X^{(i)}$ is a sequence of length $m_i$ on key space of size $n_i$). 
We consider a natural way to compose these sequences by using a sequence  
$\tilde X \in [\ell]^{m}$ as a template, where $m = \sum_i m_i$ and $n= \sum_i n_i$.
The ``composed sequence'' $S=(S_1,\ldots, S_m) \in [n]^m$ is defined as $S_t = X^{(\tilde X_t)}_{\sigma(t)} + N_t$, where $N_t = \sum_{i=1}^{\tilde{X}_t - 1} n_i$, and $\sigma(t) = |\set{t' \leq t: \tilde X_{t'} = \tilde X_t}|$.
Denote the composed sequence by $X= \tilde X\{X_1,\ldots, X_\ell\}$. Intuitively, the sequences are interleaved spatially, and the order in which they produce the next element of the composed sequence is governed by the template sequence. See Figure~\ref{fig:interleave} for illustration.

\vspace{-0.05in}
\begin{theorem} 
Let 
$X = \tilde X\{X_1,\ldots, X_\ell\}$, and let $\tilde f, f_1,\ldots, f_\ell$ be sound BST bounds. Then 
$f(X) = \tilde f(\tilde X) + \sum_{i=1}^{\ell} f_i(X_i)$ is also a sound BST bound.
\end{theorem}   
\vspace{-0.05in}

This result allows us to analyze the optimum of natural classes of sequences whose easiness was not implied by any of the known bounds.

\smallskip
\noindent{\bf Key-independent setting.} We revisit the key-independent setting, in which Iacono showed that optimality is equivalent to working set~\cite{Iacono05}. We show the following.

\vspace{-0.05in}
\begin{theorem} 
In the key-independent case Move-to-root is optimal. 
\end{theorem}  
\vspace{-0.05in}

This may seem surprising, since Move-to-root is a rather simple heuristic, not guaranteed to achieve even sublinear amortized access time. 
The result is nevertheless consistent with intuition, since in other key-independent models (e.g.\ list update problem) heuristics similar to Move-to-root have already been known to be asymptotically optimal~\cite{list_update}. 
Moreover, in such key-independent problems, ``useful structure'' has typically been described via bounds resembling the working set bound (see e.g.~\cite{Panagiotou, albers, Albers2}), which in the key-independent BST case is indeed the ``full story'' of optimality. 

\smallskip
\noindent{\bf Relations between classical bounds.}
We make explicit the equivalence between two popular notions in BST bounds, namely, {\em information-theoretic proximity with weighted elements} (such as in weighted dynamic finger) and {\em proximity of keys in a reference tree} (such as in static optimality and lazy finger).
This equivalence has implicitly appeared several times in the literature, but it has not been applied to some of the classical BST bounds. 

By making the connection explicit, the landscape of known bounds becomes clearer. In particular, the following become obvious: (i) static optimality is just an access lemma with fixed weight function, and (ii) static finger is an unweighted version of static optimality. Using these observations, we prove some of the known properties of Splay and Greedy in a way that is arguably simpler and more intuitive than existing textbook proofs.

\smallskip

\noindent{\bf Open problems.} Some of the sound BST bounds presented in the paper are not known to be achieved by online algorithms. In particular, does any online algorithm achieve the $k$-lazy finger (times $O(\log k)$) bound when $k \geq 2$? Does any online algorithm achieve the interleave bound? Are there broad classes of linear cost sequences not captured by any of the known bounds?
Does any online algorithm achieve the bound of $O(k \log k) |X|$, where $k = m(X)$?
These questions serve as concrete intermediate steps for proving or disproving dynamic optimality. 

Our result for the decomposability parameter of a sequence is tight. The bounds for general pattern avoidance and monotone pattern parameter (defined in \textsection\,\ref{sec:dict}) are not known to be tight. 

There exist other ways of composing sequences (different from the operation used in our interleave bound). Do these operations similarly lead to composite BST bounds? In particular, if $X$ is the \emph{merge} of $X_1$ and $X_2$ (i.e.\ $X$ can be partitioned into two subsequences $X_1$ and $X_2$), does a linear cost of both $X_1$ and $X_2$ imply the linear cost of $X$? 

\begin{figure}
\begin{center}  

\includegraphics[scale=0.16]{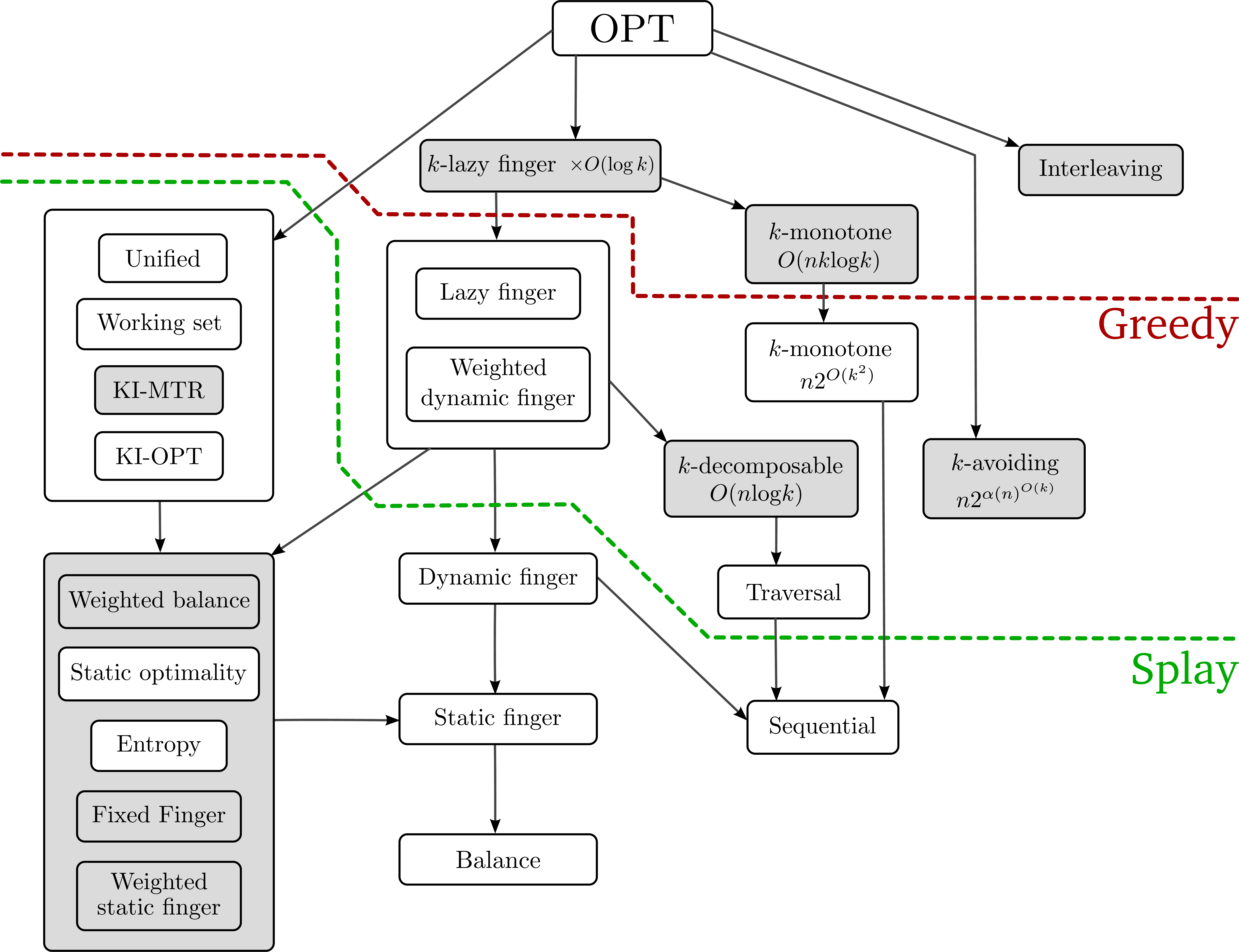}
\end{center}
\caption{\label{fig:landscape}BST bounds and relations between them. Each box represents a BST bound. Arrows indicate asymptotic domination: the source of the arrow is ``smaller'' than the target. Bounds grouped into larger boxes are asymptotically equivalent. Shaded boxes indicate bounds that are either new or strengthened in this paper (with the exception of the $k$-avoiding bound for which we prove in this paper a stronger \emph{lower bound} only). The following remarks are in order.\\
1. We define pattern-avoiding bounds in a parameterized way, such that they are well-defined for all access sequences. Let $k$ be an arbitrary positive integer. If $X$ is $k$-avoiding, then the value of the $k$-avoiding bound for $X$ is $2^{\alpha(n)^{O(k)}}$, and otherwise the value of the bound is defined to be $+\infty$. Similarly, the value of the $k$-decomposable bound is $O(n \log k)$ if $X$ is $k$-decomposable, and $+\infty$ otherwise. We define two $k$-monotone bounds, a \emph{strong}, and a \emph{weak} bound. Both bounds are set to $+\infty$ if $X$ is not $k$-monotone. In other cases, the value of the strong, respectively weak, $k$-monotone bound is $O(n k \log k)$, respectively $n 2^{O(k^2)}$. 
The traversal and sequential bounds are defined similarly: the value of the bound is $O(n)$ for a permutation sequence $X \in [n]^n$, if $X$ is a preorder traversal, respectively monotone increasing sequence, and $+\infty$ otherwise. \\
2. Boxes with a parameter $k$ indicate families of bounds: there is a different bound defined for each value of $k$. Bounds in the same box for different values of $k$ are not always comparable, i.e.\ the parameterized families of bounds are not necessarily increasing or decreasing with $k$. \\ 
3. In case of the arrows from ($k$-lazy finger $\times O(\log{k})$) to ($k$-monotone), and from the stronger ($k$-monotone) to the weaker ($k$-monotone), it is meant that the former bound with a given fixed value of $k$ dominates the latter bound \emph{with the same value $k$}.\\ 4. The arrows from (lazy finger) to ($k$-decomposable), from ($\OPT$) to ($k$-lazy finger $\times O(\log{k})$), and from ($\OPT$) to ($k$-avoiding) indicate that the former dominates the latter \emph{for all values of $k$.}\\ 5. The arrows from ($k$-monotone) to (sequential), from ($k$-decomposable) to (traversal), and from ($k$-lazy finger $\times O(\log{k})$) to (lazy finger) indicate domination for any constant $k \geq 2$.} 
\end{figure}

%% file: dictionary.tex
\section{Dictionary of BST bounds}
\label{sec:dict} 
In this section we list the BST bounds considered in the paper, marking those that are new with $\star$. Let $S=(s_{1},\dots,s_{m}) \in[n]^{m}$ be an access sequence.
A \emph{weight
	function} $w:[n]\rightarrow\mathbb{R}^{+}$ maps
elements to positive reals. For convenience, denote 
$w[i_{1}:i_{2}]=\sum_{i=\min\{i_{1},i_{2}\}}^{\max\{i_{1},i_{2}\}}w(i)$, let $W=w[1:n]$,
and $w(X)=\sum_{i\in X}w(i)$ for any set $X\subset[n]$. 
The following bounds can be defined for any access sequence $S\in[n]^{m}$.

\smallskip
\noindent{\bf Basic bounds.}

{\em Balance:} The {\em balance bound} is $B(S)=m\log n$. It describes the fact that accesses take amortized $O(\log n)$ time. We can generalize it with weights as follows.

{\em Weighted Balance$^\star$:}
For any weight function $w$, let $\AL_{w}(S)=\sum_{j=1}^{m}\log\frac{W}{w(s_{j})}$. The {\em weighted balance bound} is $\AL(S)=\min_{w}\AL_{w}(S)$. Note that for the uniform weight function $w_i = 1$, $\AL_{w}(S) = B(S)$. The reader might observe the similarity of the weighted balance bound with the \emph{access lemma}, a statement that bounds the amortized cost of a single access. The access lemma has been used to prove properties of Splay and other algorithms~\cite{ST85, Fox11, esa15tick}. 
We observe that matching the bound $\AL$ is a weaker condition than satisfying the access lemma, since here the weights are fixed
throughout the sequence of accesses, whereas the access lemma makes no such assumption.

{\em Entropy:} Let $m_{i}$ be the number of times element $i$ is accessed. The {\em entropy bound}~\cite{ST85} is $H(S)=\sum_{i=1}^{n}m_{i}\log\frac{m}{m_{i}}$.

\smallskip
\noindent{\bf Locality in keyspace.}

{\em Static Finger:} this bound depends on the distances from a fixed key.
For an arbitrary element $f \in [n]$, let $\SF_{f}(S)=\sum_{j=1}^{m}\log(|f-s_{j}|+1)$.
The {\em static finger bound}~\cite{ST85} is $\SF(S)=\min_{f}\SF_{f}(S)$. We define a weighted version as follows.

 For any weight function $w$ and any element $f$, let $\WSF_{w,f}(S)=\sum_{j=1}^{m}\log\frac{w[f:s_{j}]}{\min\{w(f),w(s_{j})\}}$. The {\em weighted static finger$^\star$ bound} is 
$\WSF(S)=\min_{w}\WSF_{w}(S)$.

{\em Dynamic Finger:} these bounds depend on the distances between consecutive accesses. The {\em dynamic finger bound}~\cite{ST85, finger1, finger2} is $\DF(S)=\sum_{j=2}^{m}\log(|s_{j}-s_{j-1}|+1)$.

	For any weight function $w$, let $\WDF_{w}(S)=\sum_{j=2}^{m}\log\frac{w[s_{j-1}:s_{j}]}{\min\{w(s_{j-1}),w(s_{j})\}}$.
	The {\em weighted dynamic finger bound}~\cite{BoseDIL14, LI16} is 
$\WDF(S)=\min_{w}\WDF_{w}(S)$.

\smallskip
\noindent{\bf Locality in time.}

{\em Working set:} 
For any $j\le m$, let the last touch time of the element $s_{j}$
be $\rho_{S}(j)=\max\{k<j\mid s_{k}=s_{j}\}$. If $j$ is the first
time that $s_{j}$ is accessed, then we set $\rho_{S}(j)=0$. 

The working set at time $j$ is defined as $w_S(j) = \{s_i \mid \rho_S{(j)} < i \leq j\}$. In words, it is the set of distinct elements accessed since the last touch time of the current element.

The {\em working set bound}~\cite{ST85} is $\WS(S)=\sum_{j=1}^{m}\log(|w_S(j)|)$.

\smallskip
\noindent{\bf Locality in a reference tree.}

{\em Static Optimality:} 
	For any fixed BST $T$ on $[n]$, let $SO_{T}(S)=\sum_{j=1}^{m}d_{T}(s_{j})$,
	where $d_{T}(s_{j})$ is the depth of $s_{j}$ in $T$ i.e.\ the distance from
	the root to $s_{j}$. The \emph{static optimality bound}~\cite{ST85} is $SO(S)=\min_{T}SO_{T}(S)$.

{\em Fixed Finger$^\star$:} 
	For any fixed BST $T$ and any element $f$, let $\FF_{T,f}(S)=\sum_{j=1}^{m}d_{T}(f,s_{j})$
	where $d_{T}(f,s_{j})$ is the distance from $f$ to $s_{j}$ in $T$.
	The \emph{fixed finger$^\star$ bound} is $\FF(S)=\min_{T,f}\FF_{T,f}(S)$.

{\em Lazy Finger:}
The previous two bounds capture the proximity of an access to the root and to a fixed key $f$ respectively. 
The lazy finger bound~\cite{BoseDIL14} captures the proximity of consecutive accesses in a reference tree.   
	For any fixed BST $T$ on $[n]$, let $\LF_{T}(S)=\sum_{j=2}^{m}d_{T}(s_{j-1},s_{j})$
	where $d_{T}(s_{j-1},s_{j})$ is the distance from $s_{j-1}$ to $s_{j}$
	in $T$. The \emph{lazy finger bound} is $\LF(S)=\min_{T}\LF_{T}(S)$.

{\em $k$-Lazy Finger$^\star$:}
We generalize the lazy finger bound to allow multiple fingers. 
Our definition 
is inspired by~\cite{BoseDIL14,DemaineILO13}. 
Let $k \in {\mathbb N}$ and $T$ be a binary search tree on $[n]$. 
A {\em finger strategy} consists of a sequence $\vec{f} \in [k]^m$ where $f_t \in [k]$ specifies the  finger that will serve the request $s_t$, and an {\em initial vector} $\vec{\ell} \in [n]^k$ where $\ell_i \in [n]$ specifies the initial location of finger $i$.
The cost of strategy $(\vec{f}$, $\vec{\ell})$ is $\LF^k_{T,\vec{f}, \vec{\ell}}(S) = \sum_{t=1}^m (1+d_T(s_t, s_{\sigma(f_t,t)}))$ where $\sigma(i,t)$ is the location of finger $i$ before time $t$, and $\sigma(i,1) = \ell_i$.  
Let $\LF^k_T(S) = \min_{\vec{f}, \vec{\ell}} \LF^k_{T,\vec{f}, \vec{\ell}} (S)$.  
In other words, for a fixed BST $T$ on key set $[n]$, $\LF^k_T(S)$ is the optimal \emph{$k$-server solution} 
that serves access sequence $S$ in tree $T$.  
We define $\LF^k(S) = \min_T \LF^k_{T}(S)$.  
It is clear form the definition that $\LF^1(S) \geq \LF^2(S) \geq \ldots \geq \LF^n(S) = m$.  

{\em Unified Bound:}
The {\em unified bound}~\cite{ST85, ElmasryFI13} computes for every access the minimum among the static finger, static optimality, and working set bounds. It is defined as $\UB = \min_{T,f}\sum_{j=1}^{m}\log(\min\{|f-s_{j}|+1, d_{T}(s_{j}), |w_S(j)|\})$. The bound should not be confused with the \emph{unified conjecture}\cite{unified}, which subsumes the working set and dynamic finger bounds but is not currently known to be achieved by $\OPT$.

The described bounds are summarized in Table~\ref{tab1} in \textsection\,\ref{apptable}. We defer the definition of key-independent bounds to \textsection\,\ref{sec:MTR}.

\smallskip
\noindent{\bf Pattern avoidance.}

Pattern avoidance bounds are, in some sense, different from the other BST bounds; they capture a more ``global'' structure, whereas other bounds all measure a broadly understood ``locality of reference''. 

The {\em pattern avoidance parameter} $p(X)$ is the smallest integer such that $X$ avoids some permutation pattern $\sigma$ of length $p(X)$. If $k \geq p(X)$, we say that $X$ is $k$-avoiding.  
The following are special cases of this parameter.

The {\em monotone pattern parameter} $m(X)$ is the smallest integer such that $X$ avoids one of the patterns $(1,\ldots, m(X))$ or $(m(X),\ldots, 1)$. If $k \geq m(X)$, we say that $X$ is $k$-monotone. The monotone pattern parameter of \emph{sequential access} is $m(X) = 2$.

The {\em decomposability parameter} is defined for permutation access sequences ($X \in [n]^n$). Parameter $d(X)$ is the smallest integer such that $X$ avoids all simple permutations of length $d(X)+1$ and $d(X)+2$. If $k \geq d(X)$, we say that $X$ is $k$-decomposable. There is an equivalent definition of $k$-decomposability in terms of a block decomposition of $X$, see~\textsection\,\ref{sec:lfapp}. For a traversal sequence $X$ (i.e.\ the preorder sequence of some BST) we have $d(X)=2$. 

We refer to~\cite{ChalermsookG0MS15} for more details on pattern-avoiding bounds. The fact that $OPT$ is linear for traversal and for sequential access is well-known. The following relations are shown in~\cite{ChalermsookG0MS15} between pattern avoidance parameters and $\opt$.

\vspace{-0.05in}

\begin{theorem}[\cite{ChalermsookG0MS15}] 
Let $X$ be a permutation input sequence in $[n]^n$, let $p = p(X)$, let $d = d(X)$, and $m = m(X)$. The following relations hold:  
\begin{compactitem} 
\item $\Greedy(X) \leq n 2^{\alpha(n)^{O(p)}}$
\item $\Greedy(X) \leq n 2^{O(d^2)}$ and $\opt(X) \leq O(m \log d)$
\item $\Greedy(X) \leq n 2^{O(m^2)}$
\end{compactitem} 
\end{theorem} 

\vspace{-0.08in}

\smallskip
\noindent{\bf Known relations between bounds.}  We refer to Figure~\ref{fig:landscape} for illustration. By definition, the weighted bounds are stronger than their unweighted counterparts because of the uniform weight function mapping all elements to $1$, therefore $\AL(S)\le B(S)$, $\WSF(S)\le \SF(S)$, $\WDF(S)\le DF(S)$ for any sequence $S$. In \cite{ChalermsookG0MS15} it is shown that pattern-avoidance bounds are incomparable with the dynamic finger and working set bounds.

\vspace{-0.05in}

\begin{theorem} The following relations hold for any sequence $S$:
\begin{compactitem}
\item \cite{mehlhorn1975nearly} $SO(S)=\Theta(H(S))$.
\item \cite{BoseDIL14} $\LF(S)=\Theta(\WDF(S))$.
\item \cite{ElmasryFI13} $\UB(S)=\Theta(\WS(S))$.
\end{compactitem}
\end{theorem}

\vspace{-0.05in}

\noindent{\bf Known competitiveness of algorithms.}
We only focus on the complexity of the Splay and Greedy algorithms. Other algorithms considered in the literature include Tango trees~\cite{tango} and Multi-splay~\cite{multisplay}. We summarize the bounds known to be achieved by Splay and Greedy in Theorem~\ref{thm:known} deferred to \textsection\,\ref{app_known}. These facts are also illustrated in Figure~\ref{fig:landscape}.

%% file: new-eq.tex
\section{New equivalences between classical bounds}  
In this section we revisit some of the classical BST bounds and establish new equivalences. 
In \textsection\,\ref{sec:access lemma class} we prove that the weighted balance, weighted static finger, static optimality, and fixed finger bounds are equivalent. As an application, a new proof is presented that Splay and Greedy satisfy these properties in \textsection\,\ref{sec:static new}. In \textsection\,\ref{sec:MTR} we show that Move-to-root is the optimal algorithm, when key values are randomly permuted.

\vspace{-0.1in}

\subsection{Static optimality and equivalent bounds}
\label{sec:access lemma class}

The conceptual message of this section is that the weighted version of information-theoretic proximity is equivalent to proximity in a reference tree. We start by stating the technical tools that allow the conversion between the two settings.

\smallskip
\noindent{\bf Weight $\Rightarrow$ Tree:}
We refer to the 
randomized construction of a BST from an arbitrary weight function due to Seidel and Aragon~\cite{SeidelA96}. 

\vspace{-0.09in}
\begin{lemma}[\cite{SeidelA96}]
\label{thm:tree-from-weight}
	Given a weight function $w$,
	there is a randomized construction of a BST $T_{w}$ with the following
	properties:
	\begin{compactitem}
		\item the expected depth of element $i$ is $E[d_{T_{w}}(i)]=\Theta(\log\frac{W}{w(i)})$,
		and 
		\item the expected distance from element $i$ to $j$ is $E[d_{T_{w}}(i,j)]=\Theta(\log\frac{w[i:j]}{\min\{w(i),w(j)\}})$.
	\end{compactitem}
\end{lemma}
\vspace{-0.05in}
When only the first property is needed, we can use a deterministic construction, see \textsection\,\ref{app1}.

\medskip
\noindent{\bf Tree $\Rightarrow$ Weight:} 
Given a tree $T$, the following assignment of weights is folklore. 

\begin{lemma}
\label{lem:subtree-sum}  Let $T$ be a BST and define $w(i) = 4^{-d_T(i)}$ for all $i \in [n]$. Then for any key
$i \in [n]$, $\sum_{j \in T_i} w(j) = \Theta(w(i))$.  
\end{lemma} 
\vspace{-0.1in}  
\begin{proof} 
From the BST property, there are at most $2^d$ nodes that are at depth $d$ in the subtree $T_i$. 
Therefore, $w(i) \leq \sum_{j \in T_i} w(j) \leq \sum_{d' = 0}^{\infty} 2^{d'} 4^{-d_T(i) - d'} \leq 2 w(i)$.
\end{proof} 

\vspace{-0.05in}  
We show that for any sequence $S$, the following bounds are equivalent: 
weighted balance, static optimality, weighted static finger, and fixed finger. We defer the proofs to \textsection\,\ref{app3}. We remark that all proofs in this section are inspired by the equivalence between the weighted dynamic finger and lazy finger bounds in \cite{BoseDIL14}.

\begin{theorem}
\label{thm:WBSO}
For all sequences $S$, we have $\WB(S) = \SO(S) = \WSF(S) = \FF(S)$ (up to constant factors).
\end{theorem}

\noindent{\bf Discussion.} We can interpret the theorem as follows: (i) Any algorithm satisfying the access lemma~\cite{ST85} obviously satisfies static optimality, because $\SO = \AL$, and $\AL$ is equivalent to the access lemma with the restriction that the weight function is fixed throughout the sequence. (ii) In static BSTs, fixing the finger at the root is the best choice up to a constant factor, because $\SO = \FF$. (iii) $\SO$ is now obviously stronger than $\SF$ because $\SO = \WSF$, and $\WSF$ is the weighted version of $\SF$.

\input{mindepth}

\subsection{Move-to-root is optimal when elements are randomly permuted}
\label{sec:MTR} 
Let $\pi:[n]\rightarrow[n]$ be a permutation. 
For any sequence $S\in[n]^{m}$, we denote by $\pi(S)=(\pi(s_{1}),\dots,\pi(s_{m}))$ the permuted sequence of $S$ by $\pi$. The \emph{key-independent optimality bound}
is defined as $\kiopt(S)=E_{\pi}[OPT(\pi(S)]$, where the expectation is over the uniform random distribution of permutations of size $n$. 

\vspace{-0.02in}
\begin{theorem}
	[\cite{Iacono05}]\label{thm:KIOPT=WS} $\kiopt(S)=\Theta(\WS(S)+n\log n)$
	for any sequence $S$%
	.\footnote{The term $n\log n$ is missing in \cite{Iacono05}.%
	}
\end{theorem}
\vspace{-0.02in}

This shows that the expected cost over a random order of the elements,
of the optimal algorithm is equivalent to the working set bound up
to a constant factor if the length of the sequence is $m\ge n\log n$. 
In this section we show that in the key-independent setting even the very simple heuristic that just rotates the accessed element to the root, is optimal. This algorithm is called Move-to-root~\cite{Allen-Munro}, and we denote its total cost for accessing a sequence $S$ from initial tree $T$ by $MTR_{T}(S)$.

\vspace{-0.02in}
\begin{definition}
	For any sequence $S$ and any initial tree $T$, the \emph{key-independent
		move-to-root bound} is $\kimtr_{T}(S)=E_{\pi}[MTR_{T}(\pi(S)]$ where
	$\pi$ is a random permutation.\end{definition}
\vspace{-0.02in}

The following theorem shows that key-independent move-to-root (starting from a balanced
tree), key-independent optimality, and working set bounds are all
equivalent when the length of the sequence is $m\ge n\log n$. The proof is deferred to \textsection\,\ref{app_MTR}. 
\vspace{-0.02in}
\begin{theorem}
	Let $T$ be a BST of logarithmic depth. Then, $\kimtr_{T}(S)=\Theta(\kiopt(S)+n\log n)=\Theta(\WS(S)+n\log n)$
	for any sequence $S$.\end{theorem}

We remark that Move-to-root is known to have another property related to randomized inputs: If accesses are drawn independently from some distribution, then Move-to-root achieves static optimality for that distribution~\cite{Allen-Munro}.

%% file: mindepth.tex
\vspace{-0.05in}
\subsection{New proofs of static optimality}
\label{sec:static new}
In this section we give a simple direct proof that Splay and Greedy
achieve static optimality. By Theorem~\ref{thm:WBSO}, this implies that the other bounds are also achieved. These facts are well-known~\cite{ST85,Fox11}, but we find the new proofs to provide additional insight. 
We present the proof for Splay and defer the proof for Greedy to \textsection\,\ref{app_greedy}. 

We use a potential function with a a clear combinatorial interpretation.

\smallskip
\noindent{\bf Min-depth potential function.}
Fix a BST $R$, called a \emph{reference tree}. Let $T$ be
the \emph{current tree} maintained by our BST algorithm (either Splay or Greedy).
Let $T(i)$ denote the set of elements in the subtree rooted at $i$. 
For each
element $i$, the potential of $i$ with respect to $R$ is $\phi_{R}(i)=-2\cdot\min_{j\in T(i)}d_{R}(j)$.
The \emph{min-depth potential} of $T$ with respect to $R$ is $\phi_{R}(T)=\sum_{i=1}^{n}\phi_{R}(i)$. We will drop the subscript $R$ for convenience. We present an easy but crucial fact.

\begin{fact}
\label{fact:unique}For any interval $[a,b]$ and any BST $R$, there
is a unique element $c\in[a,b]$ with smallest depth in $R$.\end{fact}
\vspace{-0.15in}
\begin{proof}
Suppose there are at least two elements $c$ and $c'$ with smallest
depth. Then the lowest common ancestor $LCA(c,c')$ would have smaller
depth, which is a contradiction.
\end{proof}

\vspace{-0.15in}
\begin{theorem}
The amortized cost of splay for accessing element $i$ is $O(d_{R}(i))$.\end{theorem}
\vspace{-0.15in}
\begin{proof}
\let\qed\relax
Let $\phi^{b}(i)$ and $\phi^{a}(i)$ be the potential of $i$ before
and after splaying $i$. We have $\phi^{a}(i)=0$ because $i$ is
the root, and $\phi^{b}(i)\ge-2d_{R}(i)$. 

For each zigzig or zigzag step (see~\cite{ST85} for the description of the Splay algorithm), let $x,y,z$ be the elements in the
step where 
$d_{T}(x) > d_{T}(y) > d_{T}(z)$. Let $\phi(i)$ and $\phi'(i)$
be the potential before and after the step, and let $T$ and $T'$ be the tree before and after the step. 
It suffices to prove that the cost is at most $3(\phi'(x)-\phi(x))$.
This is because by telescoping, the total cost for splaying
$i$ will be $O(\phi^{a}(i)-\phi^{b}(i))=O(d_{R}(i))$, and the  %
amortized cost in the final zig step is trivially at most $1+2(\phi'(x)-\phi(x))$. %
We analyze the two cases. 

Zigzig: we have that $T(x)$ and $T'(z)$ are such that $T(x)\cap T'(z)=\emptyset$
and $T(x),T'(z)\subset T'(x)$. By Fact~\ref{fact:unique}, either $\phi(x)+2\le\phi'(x)$
or $\phi'(z)+2\le\phi'(x)$, so we have $2\le2\phi'(x)-\phi(x)-\phi'(z)$.
Therefore, the amortized cost is
\begin{eqnarray*}
2+\phi'(x)+\phi'(y)+\phi'(z)-\phi(x)-\phi(y)-\phi(z) & = & 2+\phi'(y)+\phi'(z)-\phi(x)-\phi(y)\\
 & \le & 2+\phi'(x)+\phi'(z)-2\phi(x)\\
 & \le & (2\phi'(x)-\phi(x)-\phi'(z))+\phi'(x)+\phi'(z)-2\phi(x)\\
 & = & 3(\phi'(x)-\phi(x)).
\end{eqnarray*}

Zigzag: we have that $T'(y)$ and $T'(z)$ are such that $T'(y)\cap T'(z)=\emptyset$
and $T'(y),T'(z)\subset T'(x)$. By Fact~\ref{fact:unique}, we have $2\le2\phi'(x)-\phi'(y)-\phi'(z)$. Therefore, the amortized cost is

\vspace{-0.15in}

\begin{eqnarray*}
2+\phi'(x)+\phi'(y)+\phi'(z)-\phi(x)-\phi(y)-\phi(z) & = & 2+\phi'(y)+\phi'(z)-\phi(x)-\phi(y)\\
 & \le & (2\phi'(x)-\phi'(y)-\phi'(z))+\phi'(y)+\phi'(z)-\phi(x)-\phi(y)\\
 & \le & 2(\phi'(x)-\phi(x)). 
\end{eqnarray*}

\end{proof}
\vspace{-0.15in}
It is instructive to observe the similarity between the min-depth potential
and the \emph{sum-of-logs} potential~\cite{ST85}, which is essentially the ``soft-min'' version of min-depth, if the weights are set to $4^{-d_{R}(i)}$. Such a weight assignment is used in proving static optimality e.g.\ in~\cite{Eppstein_blog}.

%% file: lazy-finger.tex
\section{A new landscape via $k$-lazy fingers}
\label{sec:lazy-finger}
In this section we study the lazy finger bound and its generalization to multiple fingers. 
In \textsection\,\ref{sec:simulation lazy fingers} we argue that $\opt(X) \leq O(\log k) \cdot \LF^k(X)$ for all sequences $X$, refining the result of \cite{DemaineILO13}, which had an overhead factor of $O(k)$ instead of $O(\log k)$. This bound is essentially tight: We show in \textsection\,\ref{sec:separation} that there is a sequence $X$ for which $\opt(X) = \Theta(\log k) \cdot \LF^k(X)$.  

\subsection{Applications of lazy fingers}
\label{sec:lfapp}
\paragraph*{Application 1: Lazy fingers and decomposability.}
\vspace{-0.1in}
First, we give necessary definitions. 
Let $\sigma = (\sigma(1),\ldots, \sigma(n))$ be a permutation. 
For $a,b: 1 \leq a< b \leq n$, we say that $[a,b]$ is a {\em block} of $\sigma$ if $\set{\sigma(a),\ldots, \sigma(b)} = \set{c,\ldots, d}$ for some integer $c,d \in [n]$.
A {\em block partition} of $\sigma$ is a partition of $[n]$ into $k$ blocks $[a_i, b_i]$ such that $(\bigcup_i [a_i, b_i]) \cap \N = [n]$.   
For such a partition, for each $i=1,\ldots, k$, consider a permutation $\sigma_i\in S_{b_i -a_i+1}$ obtained as an order-isomorphic permutation when restricting $\sigma$ on $[a_i, b_i]$.
For each $i$, let $q_i \in [a_i, b_i]$ be a representative element of $i$.  
The permutation $\tilde \sigma \in [k]^k$ that is order-isomorphic to $\set{\sigma(q_1),\ldots, \sigma(q_k)}$ is called a {\em skeleton}  of the block partition. 
We may view $\sigma$ as a {\em deflation} $\tilde \sigma[\sigma_1, \ldots, \sigma_k]$.  

Now we provide a recursive definition of $d$-decomposable permutations. We refer to~\cite{ChalermsookG0MS15} for more details.
A permutation $\sigma$ is $d$-decomposable if 
$\sigma = (1)$, or  
$\sigma = \tilde \sigma[\sigma_1,\ldots, \sigma_{d'}]$ for some $d' \leq d$ and each permutation $\sigma_i$ is $d$-decomposable.  

\vspace{-0.05in}
\begin{lemma} Let $S$ be a $k$-decomposable permutation of length $n$. Then $\LF(S) \le 4(\abs{S} - 1) \ceil{\log k}$.
\end{lemma}
\vspace{-0.1in}
\begin{proof} 
It is sufficient to define a reference tree $T$ for which $LF_T(S)$ achieves such bound. 
We remark that the tree will have auxiliary elements.  
  We construct $T$ recursively. 
If $S$ has length one, $T$ has a single node and this node is labeled by the key in $S$. 
Clearly, $\LF_T(S) = 0$.

Otherwise, let $S = \tilde S[S_1,\ldots,S_j]$ with $j \in [k]$ the outermost partition of $S$. Denote by $T_i$ the tree for $S_i$ that has been inductively constructed.  
Let $T_0$ be a BST of depth at most $\ceil{\log j}$ and with $j$ leaves.
Identify the $i$-th leaf with the root of $T_i$ and assign keys to the internal nodes of $T_0$ such that the resulting tree is a valid BST. 
Let $r_i$ be the root of $T_i$, $0 \le i \le j$ and let $r = r_0$ be the root of $T$. 
Then

\vspace{-0.15in}
\begin{align*}
d_T(r,s_1) & \le \ceil{\log k} + d_{T_1}(r_1,s_1) \\
d_T(r,s_n) & \le  \ceil{\log k} + d_{T_j}(r_j,s_n) \\
d_T(s_{i-1},s_i) & \le  \begin{cases} d_{S_\ell}(s_{i-1},s_i)    & \text{if $s_{i-1},s_i \in S_\ell$}\\
                    2\ceil{\log k} + d_{T_\ell}(r_\ell,s_{i-1}) + d_{T_{\ell + 1}}(r_{\ell+1},s_i)  &\text{if $s_{i-1} \in S_\ell$ and $s_i \in S_{\ell + 1}$,}\end{cases}
\end{align*}
and hence 
\begin{align*}
\LF_T(S) & = d_T(r,s_0) + \sum_{i \ge 2} d_T(s_{i-1},s_i) + d_T(s_n,r) \\
&\le 2j \ceil{\log k}+ \sum_{1 \le \ell \le j} LF_{T_{\ell}}(S_{\ell}) \quad
\le \quad 2j  \ceil{\log k} + \sum_{1 \le \ell \le j} 4(\abs{S_\ell} - 1) \ceil{\log k}\\
&\le (2j - 4j + 4\sum_{1 \le \ell \le j} \abs{S_i})\ceil{\log k}\ \  \le \ \ 4(\abs{S} - 1),
\end{align*}
where the last inequality uses $j \ge 2$. 
\end{proof}

Combining this with the Iacono-Langerman result~\cite{LI16}, we conclude that 
for any $k$-decomposable sequence $S$, we have $\Greedy(S) \leq O(|S| \log k)$.  
This strengthens our earlier result~\cite{ChalermsookG0MS15} that $\Greedy(S) \leq n2^{O(k^{2})}$.

\paragraph*{Application 2: Improved relation between $\opt(S)$ and $m(S)$.}
Let $k = m(S)$, i.e.\ the smallest integer such that $S$ is $k$-monotone. 
In \cite{ChalermsookG0MS15} we show that $OPT(S)\leq \Greedy(S) \leq |S|2^{O(k^{2})}$.
Here we show the substantially stronger bound $\opt(S) \leq O(|S| k \log k)$, raising the obvious open question of whether any online BST can match this bound. 

\begin{lemma}
Let $S$ be a $k$-monotone sequence of length $n$. Then $LF^k(S)=O(n k)$.\end{lemma}
\begin{proof}
	$S$ avoids $(k+1,\ldots, 1)$ or $(1,\ldots, k+1)$. Assume that the first case holds (the argument for the other case is symmetric). 
Then, $S$ can be partitioned into $k$ subsequences  
	$S_{1},\dots,S_{k}$, each of them increasing, furthermore, such a partition can be computed online. 
We argue that $LF^k_T(S)  = O(n k)$ for any BST $T$.

Let $T$ be any binary search tree containing $[n]$ as elements. 
Consider $k$ lazy fingers $f_{1},\dots,f_{k}$
	in $T$.
We define the strategy for the fingers as follows: When $s_j$ is accessed, if $s_j \in S_i$, then move finger $f_i$ to serve this request.  
	Observe that $f_{i}$ only needs to do in-order traversal
	in $T$ (since the subsequence $S_i$ is increasing), which takes at most $O(n)$ steps. 
Thus, $LF^k_{T}(S)=O(nk)$.
\end{proof}
By \Cref{thm:mainlf}, we can simulate $k$-finger with an overhead factor of 
$O(\log k)$, concluding with the following theorem.
\begin{theorem}
Let $S$ be a $k$-monotone sequence. Then 
$OPT(S)=O(|S| \cdot k\log k)$.
\end{theorem}

%% file: interleaving.tex
\section{Combining easy sequences}
\label{sec:interleaving}
Recall that for any sequence $S\in[n]^{m}$, $\opt(S)$ denotes
the optimal cost for executing $S$ on a BST which contains $[n]$
as elements. Let $P=([a_{1},b_{1}],\dots,[a_{k},b_{k}])$ be a partitioning
of $[n]$ into $k$ intervals. That is, $a_{1}=1,b_{k}=n$ and $b_{i}=a_{i+1}-1$.
Given $P$, we can define $S_{1},\dots,S_{k}$ and $\tilde{S}$ as
follows. For each $1\le i\le k$, $S_{i}$ is obtained from $S$ by
restriction to $[a_{i},b_{i}]$. That is, for each $s_{j}\in[a_{i},b_{i}]$
starting from $j=1$ to $m$, we append $s_{j}-a_{i}+1$ to the sequence
$S_{i}$. Let $m_{i}$ be the length of $S_{i}$ and hence $\sum_{i=1}^{k}m_{i}=m$.
Next, we define $\tilde{S}=(\tilde{s}_{1},\dots,\tilde{s}_{m})\in[k]^{m}$.
For each $j\le m$, if $s_{j}\in[a_{i},b_{i}]$, then $\tilde{s}_{j}=i$. 

The main theorem of this section is the following. See \Cref{sec:interleaving_appendix} for the proof.
\begin{theorem}
[Time-interleaving Bound]\label{thm:interleaving}For any sequence
$S$ and a partition $P$ of $[n]$ into $k$ intervals, $\opt(S)\le\sum_{i=1}^{k}\opt(S_{i})+3\opt(\tilde{S})$.
\end{theorem}
This theorem bounds the optimal cost of any sequence $S$
that can be obtained by ``interleaving'' $k$ sequences $S_{1},\dots,S_{k}$
according to $\tilde{S}$. To illustrate the power of this result,
we give a bound for the optimal cost of the ``tilted grid'' sequence. 
Let $\ell = \sqrt{n}$ and consider the point set $\pset = \set{(i \ell + (j-1), j \ell + i-1): i,j \in [\ell] }.$
It is easy to check that there are no two points aligning on $x$ or $y$ coordinates, therefore $\pset$ corresponds to a permutation $S \in [n]^n$. In~\cite{ChalermsookG0MS15} it is shown that none of the known bounds imply $\opt(S)=O(n)$. We observe that the tilted grid sequence can be seen as a special case of a broad family of ``perturbed grid''-type sequences, amenable to similar analysis.

\begin{corollary}
Let $S$ be the tilted grid sequence. Then $\opt(S)=O(n)$.\end{corollary}
\begin{proof}
For $1\le i\le\sqrt{n}$, let $[a_{i},b_{i}]=[1+(i-1)\sqrt{n},i\sqrt{n}]$.
Then $S_{i}$ is a sequential access of length $\sqrt{n}$, and $\tilde{S}$
is a sequential access of length $\sqrt{n}$ repeated $\sqrt{n}$ times. So $\opt(S_{i})=O(\sqrt{n})$ and $\opt(\tilde{S})=O(n)$.
By \Cref{thm:interleaving}, $\opt(S)\le\sum_{i=1}^{\sqrt{n}}\opt(S_{i})+3\opt(\tilde{S})=O(n)$.
\end{proof}

%% file: examples.tex
\section{Separations between BST bounds}
\label{sec:separation} 
In this section we show examples that separate the bounds to the largest extent possible. For lack of space, we defer the proofs to \textsection\,\ref{sec:app_ex}. 

We first discuss the gap between $\opt$ and other BST bounds.  
When we say that a BST bound $f$ is {\em tight}, we mean that there are infinitely many sequences $X$ for which $\opt(X) \geq c f(X)$ for some constant $c$ (not depending on $X$). 
It is in this sense that many classical bounds (working set, static optimality, dynamic finger) are tight. 
We emphasize that many bounds, such as $k$-lazy finger, $k$-monotone, and $k$-avoiding, are in fact families of bounds (parameterized by $k \in {\mathbb N}$), e.g.\ the $k$-monotone bounds are given by $\set{f_k}$ where  $f_k(X)  = |X| \cdot k \log k$ if $X$ avoids $(k,\ldots, 1)$ or $(1,\ldots, k)$.
Therefore, the concept of tightness for these bounds is somewhat different. 
Our results are summarized in the following theorem.  

\begin{theorem}
\label{thm:tightness}  

For each $k$ (possibly a function that depends on $n$), there are infinitely many sequences $S_1$, $S_2$, $S_3$, for which the following holds:
\begin{compactitem}  
\item  $\LF^k(S_1) \cdot \log k\le \opt(S_1)$, 

\item 
$m(S_2) = k$, and $\opt(S_2) = \Omega (|S_2| \cdot \log k)$,

\item 
$p(S_3) = k$, and $\opt(S_3) = \Omega (|S_3| \cdot \sqrt{k})$. 
\end{compactitem}  
\end{theorem}

The results are derived using information-theoretic arguments.
For $k$-monotone bounds, this technique cannot prove $\opt(S) \geq |S| \cdot g(k)$ for a super-logarithmic function $g(k)$. 
For $k$-avoiding bounds, the information-theoretic limit is $|S| \cdot k$.

Next, we show a strong separation in the hierarchy of lazy finger bounds.
The results are summarized in the following theorem. 

\begin{theorem}
\label{thm:hierarchy}  
For any $k$ and infinitely many $n$, there is a sequence $S_k$ of length $n$, such that: 
\begin{itemize} 
\item $\LF^{k-1}(S_k) = \Omega(\frac{n}{k} \log (n/k))$ 

\item $\LF^{k}(S_k) = O(n)$ (independent of $k$) 

\item $S_k$ avoids $(k+1, k, \ldots, 1)$  

\item $\opt(S_k) = O(n)$ (independent of $k$)  
\end{itemize} 
\end{theorem} 

This theorem also implies a weak separation in the class of lazy finger bounds and monotone pattern bounds: For constant $k$, a sequence $S_k$ is linear when applying the monotone BST bound, but the lazy finger bound with $k-O(1)$ fingers would not give better than $\Omega(n \log n)$. 

Finally, we show that lazy fingers are not strong enough to subsume the classical bounds.  
\begin{theorem}
For any $k$ and infinitely many $n$, there are sequences $S_1$ and $S_2$ of length $n$, such that: 
\begin{compactitem}
\item $\WS(S_1) = o(\LF^{k}(S_1))$, and
\item $\LF^{k}(S_2) = o(\WS(S_2))$.
\end{compactitem}
\end{theorem}

%% file: appendix.tex
\section{Bounds known to be achieved by Splay and Greedy}
\label{app_known}

\begin{theorem} \label{thm:known} The following relations hold for any sequence $S$:
\begin{compactitem}
\item \cite{ST85} Splay achieves the $\WB$, $\SO$, $\SF$, $\WS$ bounds.
\item \cite{tarjan_sequential} Splay achieves the sequential bound.
\item \cite{finger1, finger2} Splay achieves the $\DF$ bound.
\item \cite{Fox11} Greedy achieves the $\WB$, $\SO$, $\SF$, $\WS$, and sequential bounds.
\item \cite{LI16} Greedy achieves the $\WDF$ bound.
\item \cite{ChalermsookG0MS15} Greedy achieves the traversal, $k$-avoiding, and weak $k$-monotone bounds.
\end{compactitem}
\end{theorem}

\newpage

\section{Table of BST bounds}
\label{apptable}

\begin{table}[h]
	\begin{tabular}{|c|c|c|}
		\hline 
		Bounds & Acronym & Formula\tabularnewline
		\hline 
		\hline 
		Balance & $B$ & $m\log n$\tabularnewline
		\hline 
		Weighted Balance & $\AL$ & $\min_{w}\sum_{j=1}^{m}\log\frac{W}{w(s_{j})}$\tabularnewline
		\hline 
		Entropy  & $H$ & $\sum_{i=1}^{n}m_{i}\log\frac{m}{m_{i}}$\tabularnewline
		\hline 
		Static Finger & $\SF$ & $\min_{f}\sum_{j=1}^{m}\log(|f-s_{j}|+1)$\tabularnewline
		\hline 
		Weighted Static Finger$\star$ & $\WSF$ & $\min_{w,f}\sum_{j=1}^{m}\log\frac{w[f:s_{j}]}{\min\{w(f),w(s_{j})\}}$\tabularnewline
		\hline
		Unified Bound & $\UB$ & $\min_{T,f}\sum_{j=1}^{m}\log(\min\{|f-s_{j}|+1, d_{T}(s_{j}), |w_S(j)|\})$ \tabularnewline
		\hline 
		Dynamic Finger & $\DF$ & $\sum_{j=2}^{m}\log(|s_{j}-s_{j-1}|+1)$\tabularnewline
		\hline 
		Weighted Dynamic Finger & $\WDF$ & $\min_{w}\sum_{j=2}^{m}\log\frac{w[s_{j-1}:s_{j}]}{\min\{w(s_{j-1}),w(s_{j})\}}$\tabularnewline
		\hline 
		Working Set & $\WS$ & $\sum_{j=1}^{m}\log(|w_S(j)|)$\tabularnewline
		\hline 
		Static Optimality & $SO$ & $\min_{T}\sum_{j=1}^{m}d_{T}(s_{j})$\tabularnewline
		\hline 
		Fixed Finger$\star$ & $\FF$ & $\min_{T,f}\sum_{j=1}^{m}d_{T}(f,s_{j})$\tabularnewline
		\hline 
		Lazy Finger & $\LF$ & $\min_{T}\sum_{j=2}^{m}d_{T}(s_{j-1},s_{j})$\tabularnewline
		\hline 
	\end{tabular}\protect\caption{\label{tab1}Summary of BST bound definitions}

\end{table}

\newpage

\section{Deterministic construction of a BST for given weights}
\label{app1}

\begin{theorem}
Given a weight function $W$, there is a deterministic construction of a BST $T_w$ such that the depth of every key $i \in [n]$ is $d_{T_w}(i) = O(\log \frac{W}{w(i)})$.  
\end{theorem}  

\begin{proof}
Let $w_1$, \ldots $w_n$ be a sequence of weights. We show how to construct a tree in which the depth of element $\ell$ is $O(\log w[1:\ell]/\min(w_1,w_\ell))$. 

For $i \ge 1$, let $j_i$ be minimal such that $w[1:j_i] \ge 2^i w_1$. Then $w[1:j_i - 1] < 2^i w_1$ and 
$w[j_{i-1} + 1 : j_i] \le 2^{i-1} w_1 + w_{j_i}$. 

Let $T_i$ be the following tree. The right child of the root is the element $j_i$. The left subtree is a tree in which element $\ell$ has depth $O(\log 2^{i-1} w_1/w_\ell)$. 

The entire tree has $w_1$ in the root and then a long right spine. The trees $T_i$ hang off the spine to the left. In this way the depth of the root of $T_i$ is $O(i)$. 

Consider now an element $\ell$ in $T_i$. Assume first that $\ell \not= j_i$. The depth is
\[ O\left(i + \log \frac{2^{i-1} w_1}{w_\ell}\right) = O\left(i + \log \frac{2^{i-1} w_1}{\min(w_1,w_\ell)}\right)= 
O\left(\log \frac{2^{i-1} w_1}{\min(w_1,w_\ell)}\right)= O\left(\frac{w[1:\ell]}{\min(w_1,w_\ell)}\right).\]

For $\ell = j_i$, the depth is 
\[ O\left(i\right) = O\left( \log \frac{2^i w_1}{w_1} \right) = O\left(\log \frac{w[1:j_i]}{\min(w_1,w_{j_i})}\right) .\]

\end{proof}

\newpage
\section{Missing proofs from \textsection\,\ref{sec:access lemma class}}
\label{app3}

\begin{theorem}
\label{thm:WBSO2}
	$\WB(S)=\Theta(\SO(S))$ for any
	sequence $S$.

\end{theorem}
\begin{proof}
	$\WB(S)=O(\SO(S))$: (note\footnote{In \cite{ST85} it is shown that $\WB(S)=O(H(S))$ and since $H(S)=\Theta(\SO(S))$, this implies that $\WB(S)=O(\SO(S))$. We include a proof for completeness.}) 
	Fix any BST $T$. It suffices to show the existence of a weight function $w$
	such that $\WB_{w}(S)=O(\SO_{T}(S))$. 
Define the weight function $w(i) = 4^{-d_T(i)}$. 
	Observe first that $W\leq 2$ from Lemma~\ref{lem:subtree-sum} applied at the root $r \in [n]$.  
This implies that for each key $i \in [n]$, we have $\log\frac{W}{w(i)}=\Theta(d_{T}(i))$.  

$\SO(S)=O(\WB(S))$:
	Fix any weight function $w$. 
It suffices to show the existence of a BST $T$ such
	that $\SO_{T}(S)=O(\WB_{w}(S))$. 
By choosing $T_w$ according to the distribution in~Lemma~\ref{thm:tree-from-weight},
	we have ${\mathbb E}[\SO_{T_{w}}(S)]={\mathbb E}[\sum_{j=1}^{m}d_{T_{w}}(s_{j})]=\sum_{j=1}^{m} {\mathbb E}[d_{T_{w}}(s_{j})]=\Theta(\sum_{j=1}^{m}\log\frac{W}{w(s_{j})})=\Theta(\WB_{w}(S))$.
	Since the expectation satisfies the bound, 
	there exists a tree $T$ satisfying the bound. 
\end{proof}
\begin{theorem}
	$\WSF(S)=O(\FF(S))$ for
	any sequence $S$.\end{theorem}
\begin{proof}
	Fix any BST $T$ and any element $f$. It suffices to show the existence of a
	weight function $w$ such that $\WSF_{w,f}(S)=O(\FF_{T,f}(S))$. In
	particular, for all $i$, we show that $\log\frac{w[f:i]}{\min\{w(f),w(i)\}}=\Theta(d_{T}(f,i))$
	when we set $w(i)=4^{-d_{T}(i)}$. Let $a=LCA(f,i)$ and hence $d_{T}(f,i)=d_{T}(f,a)+d_{T}(a,i)$.
	
	From Lemma~\ref{lem:subtree-sum} applied at node $a$, we have $w[f:i] \leq \sum_{j \in T_a} w(j) \leq 2 w(a)$. 
Therefore,  
	\begin{eqnarray*}
		\log\frac{w[f:i]}{\min\{w(f),w(i)\}} & \leq & \log\frac{2 w(a)}{\min\{w(f),w(i)\}}\\
		& \leq & \log\frac{2 \cdot 4^{-d_{T}(a)}}{\min\{4^{-d_{T}(f)},4^{-d_{T}(i)}\}}\\
		& \leq  & O(1) + 2 \cdot \max\{d_{T}(f)-d_{T}(a),d_{T}(i)-d_{T}(a)\}\\
		& = & O(1) + 2\cdot \max\{d_{T}(f,a),d_{T}(i,a)\}\\
		& \leq & O(1) + 2 \cdot d_{T}(f,i).
	\end{eqnarray*}
\end{proof}
\begin{theorem}
	$\FF(S)=O(\WSF(S))$ for
	any sequence $S$.\end{theorem}
\begin{proof}
	Fix any weight function $w$ and any element $f$. It suffices to
	show a BST $T$ such that $\FF_{T,f}(S)=O(\WSF_{w,f}(S))$. 
By invoking Lemma~\ref{thm:tree-from-weight} with weight $w$,  we have a distribution of BST $T_w$ that satisfies ${\mathbb E}[\FF_{T_{w},f}(S)]={\mathbb E}[\sum_{j=1}^{m}d_{T_{w}}(f,s_{j})]=\sum_{j=1}^{m}{\mathbb E}[d_{T_{w}}(f,s_{j})]= O (\sum_{j=1}^{m}\log\frac{w[f:s_{j}]}{\min\{w(f),w(s_{j})\}})=\Theta(\WSF_{w,f}(S))$.
	Since the expectation satisfies the bound, there must exist a tree
	$T$ satisfying the bound. 
\end{proof}

\begin{theorem}
	$\WSF(S)=\Theta(\FF(S))$ for
	any sequence $S$.\end{theorem}


So far, we have shown the equivalences between two pairs $\WB \Leftrightarrow \SO$ and $FF \Leftrightarrow \WSF$. 
The following two (very simple) connections complete the proof. 
\begin{theorem}
	$\FF(S)=O(\SO(S))$ for any
	sequence $S$.\end{theorem}
\begin{proof}
	Fix any tree $T$. We set $f$ to be the root of $T$, then $\FF_{T,f}(S)=\SO_{T}(S)$.\end{proof}
\begin{theorem}
$\WB(S)=O(\FF(S))$ for any sequence
	$S$. \end{theorem}
\begin{proof}
	Fix any weight function $w$ associated with the weighted static finger
	bound. It suffices to show a weight function $w$ such that $\WB_{w}(S)=O(\FF_{T,f}(S))$.
	In particular, for all $i$, we show that $\log\frac{W}{w(i)}=\Theta(d_{T}(f,i))$
	when we set $w(i)=4^{-d_{T}(f,i)}$. This follows because the number
	of elements with distance $d$ from $i$ is at most $3^{d}$, and
	so $W=\sum_{i=1}^{n}w(i)\le\sum_{d<\infty}3^{d}\cdot4^{-d}=O(1)$.
\end{proof}

\newpage
\section{The static optimality of Greedy}
\label{app_greedy}

In the geometric view where Greedy is defined~\cite{DHIKP09}, the role of the subtree $T(i)$ rooted at $i$ is played be the \emph{neighborhood} $T(i)$ of element $i$. The definition of neighborhood can be found in \cite{esa15tick,Fox11}.

\begin{theorem}
The amortized cost of Greedy for accessing element $i$ is $O(d_{R}(i))$.\end{theorem}
\begin{proof}
Suppose that Greedy, for accessing $i$, touches $x_{-l},x_{-l+1},\dots,x_{0}=i,x_{1},\dots,x_{r}$.
We only bound the amortized cost for touching $x_{0},x_{1},\dots,x_{r}$.
The analysis for the other elements is symmetric. Let $\phi(x)$ and $\phi'(x)$
be the potential of an element $x$ before, respectively after accessing $i$. Define $\phi(x_i) = \phi'(x_i) = 0$ for $i > r$. 
Note that $\phi(i),\phi'(i)\le0$
for all $i$. Also, for any $j\in T(i)$, $\phi(j)\le\phi(i)$.
The
crucial observation is that, for any $i\le r-3$, 
\[
2+\phi'(x_{i})+\phi'(x_{i+2})\le2\phi(x_{i+3}).
\]
This follows from Fact~\ref{fact:unique} and the facts that $T'(x_{i})\cap T'(x_{i+2})=\emptyset$
and $T'(x_{i}),T'(x_{i+2})\subset T'(x_{i+3})$. Therefore, the amortized
cost is
\begin{eqnarray*}
\sum_{0\le i\text{\ensuremath{\le}}r}1+\phi'(x_{i})-\phi(x_{i}) & = & \sum_{0\le i\text{\ensuremath{\le}}r-3,i\bmod4=0}4+\phi'(x_{i})+\phi'(x_{i+1})+\phi'(x_{i+2})+\phi'(x_{i+3})-\sum_{0\le i\text{\ensuremath{\le}}r}\phi(x_{i})\\
 & = & \sum_{0\le i\text{\ensuremath{\le}}r-3,i\bmod4=0}(2+\phi'(x_{i})+\phi'(x_{i+2}))+(2+\phi'(x_{i+1})+\phi'(x_{i+3}))-\sum_{0\le i\text{\ensuremath{\le}}r}\phi(x_{i})\\
 & \le & \sum_{0\le i\text{\ensuremath{\le}}r-3,i\bmod4=0}2\phi(x_{i+3})+2\phi(x_{i+4})-\sum_{0\le i\text{\ensuremath{\le}}r}\phi(x_{i})\\
 & \le & \sum_{0\le i\text{\ensuremath{\le}}r-3,i\bmod4=0}\phi(x_{i+3})+\phi(x_{i+4})+\phi(x_{i+5})+\phi(x_{i+6})-\sum_{0\le i\text{\ensuremath{\le}}r}\phi(x_{i})\\
 & \le & -\phi(x_{0})-\phi(x_{1})-\phi(x_{2})\\
 & \le & -3\phi(x_{0})\\
 & = & O(d_{R}(x_{0})).
\end{eqnarray*}
\end{proof}

\newpage

\section{Missing proofs from \textsection\,\ref{sec:MTR}}
\label{app_MTR}

\begin{lemma}[\cite{Allen-Munro}] Let $i < k$. Then $i$ is an ancestor of $k$ after serving $S$ if either $i$ was accessed and there was no access to a key in $[i+1,\ldots,k]$ after the last access to $i$ or $i$ and $k$ were not accessed and $i$ is an ancestor of $k$ in the initial tree.
\label{lem:AM} 
\end{lemma}

This implies that we obtain the same final tree if we delete all but the last access to each element from the access sequence.

We use the following property of the Move-to-root algorithm. 

\begin{fact}
Let $S = (s_1,\ldots, s_m)$ be an access sequence. 
Before accessing $s_j$, for any $j' <j$, the keys $s_{j'},s_{j'+1},\ldots, s_{j}$ form a connected component containing the root of $T'$. 
\label{prop:move-to-root structure}  
\end{fact}  

Now we are ready to prove the main theorem of the section. 

\begin{theorem}
	\label{thm:KIMTR=WS} For any sequence $S$ and any initial tree
	$T$, $\kimtr_{T}(S)=O(\WS(S)+f(n))$ where $f(n)=\sum_{i=1}^{n}d_{T}(i)$
	depends only on $n$ (and not on $m$).\end{theorem}
\begin{proof}
We first ananalyze the cost of non-first accesses $s_{\pi(j)}$. 
Consider the accesses $s_{\rho_{S}(j)},s_{\rho_{S}(j)+1},\dots,s_{j-1}$ 
where $s_{\rho_S(j)} = s_j$, and $s_{\ell} \neq s_{\ell'}$ for all $\ell, \ell' \in \{\rho_S(j), \ldots, j-1\}$: These elements are stored in a connected subtree $T'$ containing the root and $T'$ is also formed if we deleted all but the last access to any element from this sequence.  
More formally, let $J \subseteq \set{\rho_S(j), \ldots, j-1}$ be the set indices $i$ for which $s_i$ is the last appearance of its key in $w_S(j)$.
So we have $|J| = |w_S(j)|$.   
The expected access cost of $\pi(s_{j})$ is exactly the expected number of its ancestors in $T'$:  
\[\sum_{\ell \in J} Pr_\pi[\mbox{ $\pi(s_{\ell})$ is ancestor of $\pi(s_{\rho_S(j)})$}] \]

If $\pi(s_\ell) > \pi(s_j)$, then the key $\pi(s_{\ell})$ is an ancestor of $\pi(s_{j})$ only if $\pi(s_{\ell}) = \min \set{\pi(s_{\ell'})}_{\ell' \in J, \ell' \geq \ell}$ (due to Lemma~\ref{lem:AM}). 
Otherwise, if $\pi(s_{\ell}) < \pi(s_j)$, it is an ancestor only if $\pi(s_\ell) = \max \set{\pi(s_{\ell'})}_{\ell' \in J, \ell' \geq \ell}$.  
In any case, this probability is exactly $\displaystyle\frac{1}{|J \cap [\ell,j]|}$. 
Therefore, the expected access cost of $\pi(s_j)$ is at most $\sum_{\ell\in J} \displaystyle\frac{1}{|J \cap [\ell, j]|} = O(\log |J|) = O(\log |w_S(j)|)$.  

Consider next a first access $s_{j}$. By~\Cref{prop:move-to-root structure}, the elements $s_{1},s_{2},\dots,s_{j-1}$ form a connected subtree $T'$ containing the root. An argument similar to the one above shows that the expected depth of $T'$ is $O(\log |w_S(j)|)$. 
So the expected length of the search path of $s_{j}$
	is at most $O(\log |w_S(j)|)+d_{T}(s_{j}))$.

Thus, the total cost is $\sum_{j=1}^{m}O(\log |w_S(j)|)+O(\sum_{i=1}^{n}d_{T}(i))=O(\WS(S)+f(n))$.
\end{proof}
It follows that key-independent Move-to-root (starting from a balanced
tree), key-independent optimality, and working set bounds are all
equivalent when the length of the sequence is $m\ge n\log n$.
\begin{corollary}
	Let $T$ be a BST of logarithmic depth. Then, $\kimtr_{T}(S)=\Theta(\kiopt(S)+n\log n)=\Theta(\WS(S)+n\log n)$
	for any sequence $S$.\end{corollary}
\begin{proof}
	We trivially have $\kiopt(S)\le\kimtr_{T}(S)$. By \Cref{thm:KIMTR=WS,thm:KIOPT=WS}, we have $\kimtr_{T}(S)=O(\WS(S)+n\log n)=O(\kiopt(S)+n\log n)$.
	This sandwiches the quantities.
\end{proof}

%% file: lazy-finger_appendix.tex
\newpage 

\section{Proofs from \Cref{sec:lazy-finger}}
\subsection{Simulating $k$-Lazy fingers}
\label{sec:simulation lazy fingers} 
\begin{theorem}
	\label{thm:simulate k finger}For any sequence $S$, $OPT(S)\le O(\log k) \cdot LF^k(S)$.\end{theorem}

To prove this theorem, we refine the result in~\cite{DemaineILO13}
which shows how to simulate a $k$-finger BST using a standard BST.
In a $k$-finger BST
there are $k$ pointers, and each of these can move to
its parent or to one of its children. Moreover,
the position of a finger is maintained while performing rotations. 

When we execute an access sequence using $k$-finger BST, each finger is
initially at the root. 
We specify an initial tree. 
For each access, one of the $k$ fingers
must move to the accessed element. After each access the fingers remain in their position,
instead of moving to the root, as in the standard BST model.

The cost for accessing a sequence is the
total number of finger moves and rotations. 

Given an online $k$-finger algorithm $A_{kBST}$, we define an online BST algorithm $A_{sim}(A_{kBST})$ that \emph{simulates} $A_{kBST}$. In \cite{DemaineILO13} it is shown that this can be achieved with a factor $O(k)$ increase in cost. We strengthen the result by showing the following result.

\begin{theorem}
	\label{thm:overhead log k}There is a BST algorithm for simulating a
	$k$-finger BST with overhead factor $O(\log k)$.
\end{theorem}

From Theorem~\ref{thm:overhead log k} we immediately get \Cref{thm:simulate k finger} 
as a corollary.
	This is because the model which
	defines the $k$-lazy finger bound $LF^k(\cdot)$ is exactly $k$-finger
	BST except that the tree is static. 
	$LF^k(S)$ is the cost of
	$k$-finger BST on the optimal static tree. If we can simulate any
	$k$-finger BST algorithm with overhead $O(\log k)$, then we can
	indeed simulate any $k$-finger BST algorithm on any static tree,
	including the optimal tree. Therefore, $OPT(S)\le LF^k(S) \cdot O(\log k)$
	for any sequence $S$.

The rest of this section is devoted to the proof of \Cref{thm:overhead log k}.
We use the approach from \cite{DemaineILO13}. We are simulating a $k$-finger BST $T$ using a standard BST $T'$. The ingredients of the proof are: (1) To make sure that each element with a finger
on it in $T$ has depth at most $O(\log k)$ in $T'$. (In \cite{DemaineILO13}, each finger may have depth up to $O(k)$ in $T'$.) (2) To implement a deque data structure within $T'$ 
so that each finger in $T$ can move to any of its neighbors with cost
$O(\log k)$ amortized. (In \cite{DemaineILO13}, this cost is $O(1)$ amortized.) 
Therefore, to move a finger $f$ to its neighbor $x$ in $T$, we
can simply access $f$ from the root of $T'$ in $O(\log k)$ steps,
and then move $f$ to $x$ in $T'$ in $O(1)$ amortized steps. Hence,
the overhead factor is $O(\log k)$.

\subsubsection{DequeBST}
\label{sub:DequeBST}

We describe how to implement a deque in the BST model, called \emph{dequeBST}.
This is be used later in the simulation. The following construction
appears to be folklore, and it is the same as the one used in \cite{DemaineILO13}. As we have not found an explicit description in the literature, we include it here for completeness. 
\begin{lemma}
	\label{prop:DequeBST}The minimum and maximum element 
	from a dequeBST can be deleted in $O(1)$ amortized operations.
\end{lemma}
\begin{proof} The simulation is inspired by the well-known simulation of a queue by two stacks with constant amortized time per operation (\cite[Exercise 3.19]{Mehlhorn-Sanders:Toolbox}). We split the deque at some position (determined by history) and put the two parts into structures that allow us to access the first and the last element of the deque. It is obvious how to simulate the deque operations as long as the sequences are non-empty.  When one of the sequences becomes empty, we split the other sequence at the middle and continue with the two parts.  A simple potential function argument shows that the amortized cost of all deque operations is constant. Let $\ell_1$ and $\ell_2$ be the length of the two sequences, and define the potential $\Phi = |\ell_1 - \ell_2|$. As long as neither of the two sequences are empty, for every insert and delete operation both the cost and the change in potential are $O(1)$. If one sequence becomes empty, we split the remaining sequence into two equal parts. The decrease in potential is equal to the length of the sequence before the splitting (the potential is zero after the split). The cost of splitting is thus covered by the decrease of potential.

The simulation by a BST is easy. We realize both sequences by chains attached to the root. The right chain contains the elements in the second stack with the top element as the right child of the root, the next to top element as the left child of the top element, and so on.
\end{proof}

\subsubsection{Extended Hand}

To describe the simulation precisely, we borrow terminology from~\cite{DemaineILO13}. Let $T$ be a BST with 
a set $F$ of $k$ fingers $f_{1},\dots,f_{k}$. For convenience
we assume the root of $T$ to be one of the fingers. Let $S(T,F)$ be
the Steiner tree with terminals $F$. A \emph{knuckle} is a connected
component of $T$ after removing $S(T,F)$. Let $P(T,F)$ be the union
of fingers and the degree-3 nodes in $S(T,F)$. We call $P(T,F)$
the set of \emph{pseudofingers}. A \emph{tendon} $\tau_{x,y}$ is
the path connecting two pseudofingers $x,y\in P(T,F)$ (excluding $x$ and $y$) such
that there is no other $z\in P(T,F)$ inside. We assume that $x$
is an ancestor of $y$. 

The next definitions are new. For each tendon $\tau_{x,y}$, there
are two \emph{half tendons,} $\tau_{x,y}^{<},\tau_{x,y}^{>}$ containing
all elements in $\tau_{x,y}$ which are less than $y$ and greater than
$y$ respectively. Let $H(T,F)=\{\tau_{x,y}^{<},\tau_{x,y}^{>}\mid\tau_{x,y}$
is a tendon$\}$ be the set of all half tendons. 

For each $\tau\in H(T,F)$, we can treat $\tau$ as an interval $[\min(\tau),\max(\tau)]$
where $\min(\tau),\max(\tau)$ are the minimum and maximum elements
in $\tau$ respectively. For each $f\in P(T,F)$, we can treat $f$
as an trivial interval $[f,f]$. 

Let $E(T,F)=P(T,F)\cup H(T,F)$ be the set of intervals defined by
all pseudofingers $P(T,F)$ and half tendons $H(T,F)$. We call $E(T,F)$
an \emph{extended hand}\footnote{In \cite{DemaineILO13}, they define a hand which involves only the pseudofingers.}.
Note that when we treat $P(T,F)\cup H(T,F)$ as a set of elements,
such a set is exactly $S(T,F)$. So $E(T,F)$ can be viewed as a partition
of $S(T,F)$ into pseudofingers and half-tendons.

We first state two facts about the extended hand.
\begin{lemma}
	\label{lem:hand size}Given any $T$ and $F$ where $|F|=k$, there
	are $O(k)$ intervals in $E(T,F)$.\end{lemma}
\begin{proof}
	Note that $|P(T,F)|\le2k$ because there are $k$ fingers and there
	can be at most $k$ nodes with degree 3 in $S(T,F)$. Consider the
	graph where pseudofingers are nodes and tendons are edges. That graph
	is a tree. So $|H(T,F)|=O(k)$ as well. \end{proof}
\begin{lemma}
	\label{lem:hand ordered}Given any $T$ and $F$, all the intervals
	in $E(T,F)$ are disjoint.\end{lemma}
\begin{proof}
	Suppose that there are two intervals $\tau,x\in E(T,F)$ that intersect
	each other. One of them, say $\tau$, must be a half tendon. Because
	the intervals of pseudofingers are of length zero and they are distinct,
	they cannot intersect. We write $\tau=\{t_{1},\dots,t_{k}\}$ where
	$t_{1}<\dots<t_{k}$. Assume w.l.o.g.\ that $t_{i}$ is an ancestor
	of $t_{i+1}$ for all $i<k$, and so $t_{k}$ is an ancestor of a
	pseudofingers $f$ where $t_{k}<f$. 
	
	Suppose that $x$ is a pseudofinger and $t_{j}<x<t_{j+1}$ for some
	$j$. Since $t_{j}$ is the first left ancestor of $t_{j+1}$, $x$
	cannot be an ancestor of $t_{j+1}$ in $T$. So $x$ is in the left
	subtree of $t_{j+1}$. But then $t_{j+1}$ is a common ancestor of
	two pseudofingers $x$ and $f$, and $t_{j+1}$ must be a pseudofinger
	which is a contradiction.
	
	Suppose next that $x=\{x_{1},\dots,x_{\ell}\}$ is a half tendon where
	$x_{1}<\dots<x_{\ell}$. We claim that either $[x_{1},x_{\ell}]\subset[t_{j},t_{j+1}]$
	for some $j$ or $[t_{1},t_{k}]\subset[x_{j'},x_{j'+1}]$ for some
	$j'$. Suppose not. Then there exist two indices $j$ and $j'$ where
	$t_{j}<x_{j'}<t_{j+1}<x_{j'+1}$. Again, $x_{j'}$ cannot be an ancestor
	of $t_{j+1}$ in $T$, so $x_{j'}$ is in the left subtree of $t_{j+1}$.
	We know either $x_{j'}$ is the first left ancestor of $x_{j'+1}$
	or $x_{j'+1}$ is the first right ancestor of $x_{j'}$. If $x_{j'}$
	is an ancestor of $x_{j'+1}$, then $x_{j'+1}<t_{j+1}$ which is a
	contradiction. If $x_{j'+1}$ is the first right ancestor of $x_{j'}$,
	then $t_{j+1}$ is not the first right ancestor of $x_{j'}$ and hence
	$x_{j'+1}<t_{j+1}$ which is a contradiction again. Now suppose w.l.o.g.
	$[x_{1},x_{\ell}]\subset[t_{j},t_{j+1}]$. Then there must be another
	pseudofinger $f'$ in the left subtree of $t_{j+1}$, hence $\tau$
	cannot be a half tendon, which is a contradiction. 
\end{proof}

\subsubsection{The structure of the simulating BST}
\label{sub:structure of extended hand}

In this section, we describe the structure of the BST $T'$ that we
maintain given a $k$-finger BST $T$ and the set of
fingers $F$. 

For each half tendon $\tau\in H(T,F)$, let $T'_{\tau}$ be the tree
with $\min(\tau)$ as a root which has $\max(\tau)$ as a right child.
$\max(\tau)$'s left child is a subtree containing the remaining elements
$\tau\setminus\{\min(\tau),\max(\tau)\}$. We implement a \emph{dequeBST
}on this subtree as defined in \textsection\,\ref{sub:DequeBST}. By Lemma~\ref{lem:hand ordered},
intervals in $E(T,F)$ are disjoint and hence they are totally ordered.
Since $E(T,F)$ is an ordered set, we can define $T'_{E_{0}}$ to
be a balanced BST such that its elements correspond to elements in
$E(T,F)$. Let $T'_{E}$ be the BST obtained from $T'_{E_{0}}$
by replacing each node $a$ in $T'_{E_{0}}$ that corresponds to a half
tendon $\tau\in H(T,F)$ by $T'_{\tau}$. That is, suppose that the
parent, left child, and right child are $a_{u},a_{l}$ and $a_{r}$ respectively.
Then the parent in $T'_{E}$ of the root of $T'_{\tau}$ which is
$\min(\tau)$ is $a_{up}$. The left child in $T'_{E}$ of $\min(\tau)$
is $a_{l}$ and the right child in $T'_{E}$ of $\max(\tau)$ is $a_{r}$. 

The BST $T'$ has $T'_{E}$ as its top part and each knuckle of $T$
hangs from $T'_{E}$ in a determined way. 
\begin{lemma}
	\label{lem:depth pseudofinger}Each element corresponding to pseudofinger
	$f\in P(T,F)$ has depth $O(\log k)$ in $T'_{E}$, and hence in $T'$.\end{lemma}
\begin{proof}
	By Lemma~\ref{lem:hand size}, $|E(T,F)|=O(k)$. So the depth of $T'_{E_{0}}$
	is $O(\log k)$. For each node $a$ corresponding to a pseudofinger $f\in P(T,F)$,
	observe that the depth of $a$ in $T'_{E}$ is at most twice the depth
	of $a$ in $T'_{E_{0}}$ by the construction of $T'_{E}$. 
\end{proof}

\subsubsection{The cost for simulating the $k$-finger BST}

We finally prove \Cref{thm:overhead log k}. That
is, we prove that whenever one of the fingers in a $k$-finger
BST $T$ moves to its neighbor or rotates, we can update the maintained
BST $T'$ to have the structure as described in the last section with
cost $O(\log k)$.

We state two observations which follow from the structure of our maintained BST $T'$ described in \textsection\,\ref{sub:structure of extended hand}. The first observation follows immediately from Lemma~\ref{prop:DequeBST}.
\begin{fact}
	\label{lem:update tendon}For any half tendon $\tau\in H(T,F)$, we
	can insert or delete the minimum or maximum element in $T'_{\tau}$
	with cost $O(1)$ amortized.\end{fact}
Next, it is convenient to define a set $A$, called \emph{active set}, as a set
of pseudofingers, the roots of knuckles whose parents are pseudofingers,
and the minimum or maximum of half tendons. 
\begin{fact}
	\label{lemma:change exended hand}When a finger $f$ in a $k$-finger
	BST $T$ moves to its neighbor or rotates with its parent, the extended
	hand $E(T,F)=P(T,F)\cup H(T,F)$ is changed as follows. 
	\begin{enumerate}
		\item There are at most $O(1)$ half tendons $\tau\in H(T,F)$ whose elements
		are changed. Moreover, for each changed half tendon $\tau$, either
		the minimum or maximum is inserted or deleted. The inserted or deleted
		element $a$ was or will be in the active set $A$.
		\item There are at most $O(1)$ elements added or removed from $P(T,F)$.
		Moreover, the added or removed elements were or will be in the active
		set $A$.
	\end{enumerate}
\end{fact}

\begin{lemma}
	\label{lem:active set}Let $a\in A$ be an element in the active
	set. We can move $a$ to the root with cost $O(\log k)$ amortized.
	Symmetrically, the cost for updating the root $r$ to become some
	element in the active set is $O(\log k)$ amortized.\end{lemma}
\begin{proof}
	There are two cases. If $a$ is a pseudofinger or a root of a knuckle
	whose parent is pseudofinger, we know that the depth of $a$ was $O(\log k)$
	by Lemma~\ref{lem:depth pseudofinger}. So we can move $a$ to root with cost
	$O(\log k)$. Next, if $a$ is the minimum or maximum of a half tendon
	$\tau$, we know that the depth of the root of the subtree $T'_{\tau}$
	is $O(\log k)$. Moreover, by \Cref{lem:update tendon}, we can delete
	$a$ from $T'_{\tau}$ (make $a$ a parent of $T'_{\tau}$) with cost
	$O(1)$ amortized. Then we move $a$ to root with cost $O(\log k)$
	worst-case. The total cost is then $O(\log k)$ amortized. The proof
	for the second statement is symmetric.
\end{proof}
\begin{lemma}
	\label{lem:update cost}When a finger $f$ in a $k$-finger BST $T$
	moves to its neighbor or rotates with its parent, the BST $T'$ can
	be updated accordingly with cost $O(\log k)$ amortized.\end{lemma}
\begin{proof}
	According to \Cref{lemma:change exended hand}, we separate our cost
	analysis into two parts. 
	
	For the fist part, let $a\in A$ be the element to be inserted into
	a half tendon $\tau$. By Lemma~\ref{lem:active set}, we move $a$ to root
	with cost $O(\log k)$ and then insert $a$ as a minimum or maximum
	element in $T'_{\tau}$ with cost $O(\log k)$. Deleting $a$ from some half tendon with cost $O(\log k)$ is symmetric.
	
	For the second part, let $a\in A$ be the element to be inserted into
	a half tendon $\tau$. By Lemma~\ref{lem:active set} again, we move $a$
	to root and move back to the appropriate position in $T'_{E_{0}}$
	with cost $O(\log k)$. We also need rebalance $T'_{E_{0}}$ but this
	also takes cost $O(\log k)$.
\end{proof}
\begin{proof}[Proof of \Cref{thm:overhead log k}]
	We describe the simulation algorithm $A_{sim}$ with overhead $O(\log k)$.
	Let $A_{kBST}$ be an arbitrary algorithm for the $k$-finger BST
	$T$. Whenever there is an update in the $k$-finger BST $T$ (i.e.\
	a finger moves to its neighbor or rotates), $A_{sim}(A_{kBST})$ updates
	the BST $T'$ according to Lemma~\ref{lem:update cost} with cost $O(\log k)$
	amortized. $T'$ is maintained so that its structure is as described
	in \textsection\,\ref{sub:structure of extended hand}. By Lemma~\ref{lem:depth pseudofinger},
	we can access any finger $f$ of $T$ from the root of $T'$ with
	cost $O(\log k)$. Therefore, the cost of $A_{sim}(A_{kBST})$ is
	at most $O(\log k)$ times the cost of $A_{kBST}$.
\end{proof}

\subsection{Lazy finger bounds with auxiliary elements }

Recall that $LF(S)$ is defined as the minimum over all BSTs $T$ over $[n]$ of $LF_T(S)$. 
It is convenient to define a slightly stronger lazy finger bound that also allows {\em auxiliary elements}. 
Define $\widehat{LF}(S)$ as the minimum over all binary search trees $T$ that contains the keys $[n]$ (but the size of $T$ can be much larger than $n$).
We define $\widehat{LF}^k(S)$ as the $k$-lazy finger bound when the tree is allowed to have auxiliary elements.  
We argue that the two definitions are equivalent. 
  
\begin{theorem}
For any integer $k$, 
	$LF^k(S)=\Theta(\widehat{LF}^k(S))$ for all $S$.\end{theorem}
\begin{proof}
	It is clear that $\widehat{LF}^k(S)\le LF^k(S)$. 
We only need to show the converse.  

	Let $T$ be the binary search tree (with auxiliary elements) such that $LF^k_{T}(S)=\widehat{LF}^k(S)$. Denote by $\vec{f}$ the optimal finger strategy on $T$. 
	Let $[n]\cup X$ be the elements of $T$ where $X$ is the set of
	auxiliary elements in $T$. 
For each $a\in[n]\cup X$, let $d_{T}(a)$
	be the depth of key $a$ in $T$, and let $w(i)=4^{-d_T(i)}$. 
For any
	two elements $i$ and $j$ and set $Y \subseteq [n] \cup X$, let $w_{Y}[i:j]$ be the
	sum of the weight $\sum_{k \in Y\cap [i,j]} w(i)$. 
	For any $i,j\in[n]\cup X$ such that $i\leq j$, we have  
	\[
	\lg\frac{w_{[n]\cup X}[i:j]}{\min(w(i),w(j))}=O(d_{T}(i,j)),
	\]
	where $d_{T}(i,j)$ is the distance from $i$ to $j$ in $T$. So,
	this same bound also holds when considering only keys in $[n]$. That is, for $i,j \in [n]$, we have  \[
	\lg\frac{w_{[n]}[i:j]}{\min(w(i),w(j))}=O(d_{T}(i,j)).
	\]

Given the weight of $\{w(a)\}_{a\in[n]}$, the BST $T'$ (without auxiliary elements) is constructed by invoking Lemma~\ref{thm:tree-from-weight}. 
We bound the term $LF^k_{T'}(S)$ (using strategy $\vec{f}$) by 	
\[
	O( \sum_{t} d_{T'}(s_{\sigma(f_t,t)}, s_t)) = O(\sum_{t=1}^{m-1}\lg\frac{w_{[n]}[s_{t}:s_{\sigma(f_t,t)}]}{\min(w(s_{i}),w(s_{\sigma(f_t,t)}))})=O(\sum_{t=1}^{m-1}d_{T}(s_{\sigma(f_t,t)},s_{t}))=O(LF^k_{T}(S))
	\]
	where $S=(s_{1},\dots,s_{m})$. Therefore, $LF^k(S)\le LF^k_{T'}(S)=O(LF^k_{T}(S))=O(\widehat{LF}^k(S))$.
\end{proof}

%% file: interleaving_appendix.tex
\newpage 
\section{Proofs from \Cref{sec:interleaving}}
\label{sec:interleaving_appendix}
\subsection{Auxiliary Elements}

For any $X\subset\mathbb{Q}$ finite subset of rational numbers,
we let $\opt_{X}(S)$ be the optimal cost for executing $S$ on a BST
which contains $[n]\cup X$ as elements. We call $X$ \emph{auxiliary
elements}. Let $\opt_{aux}(S)=\min_{X}\opt_{X}(S)$. The following
theorem shows that auxiliary elements never improve the optimal
cost.
\begin{theorem}
\label{thm:aux equiv}$\opt(S)=\opt_{X}(S)$ for any sequence $S$ and
auxiliary elements $X$.\end{theorem}
\begin{proof}
It is clear that $\opt_{X}(S)\le \opt(S)$, because we can start with
the initial tree such that no element in $X$ is above any element
in $[n]$, and then use an optimal algorithm on a BST with $[n]$
as elements to arrange the tree without touching any element in $X$. 

Next, to show that $\opt(S)\le \opt_{X}(S)$, we will prove that for
any $Y\subset\mathbb{Q}$ and another number $y\in\mathbb{Q}\setminus Y$,
$\opt_{Y}(S)\le \opt_{Y\cup\{y\}}(S)$. It suffices to prove that, given
an optimal BST algorithm $A$ for executing $S$ on a BST $T_{A}$
with $[n]\cup Y\cup\{y\}$ as elements, we can obtain another BST
algorithm $B$ on a BST $T_{B}$ with $[n]\cup Y$ as elements whose
cost for execute $S$ is at most the cost of $A$ which is $\opt_{Y\cup\{y\}}(S)$. 

Let $y'\in[n]\cup Y$ be either the predecessor or successor of $y$.
At each time, the algorithm $B$ arranges the elements such that the
depth of $y'$ is $d_{T_{B}}(y')=\min\{d_{T_{A}}(y'),d_{T_{A}}(y)\}$
and the relative depth of other elements $[n]\cup Y\setminus\{y'\}$
in $T_{B}$ are equivalent as the corresponding elements in $T_{A}$.
Observe that, for any element $i\in[n]\cup Y$, the search path of
$i$ in $T_{B}$ is a subset of the search path of $i$ in $T_{A}$.
So the cost of $B$ at any access is at most of the cost of $A$.
This conclude the proof.
\end{proof}

\subsection{Proof of Interleaving Bound}

Before we can prove \Cref{thm:interleaving}, we need one lemma.
\begin{lemma}
\label{lem:integer as leaves}For any sequence $S\in[n]^{m}$,
there is a BST algorithm $A$ on a BST $T_{A}$ with $[n]\cup X$
as elements where $X$ is some set of auxiliary elements and all elements
in $[n]$ are maintained as leaves in $T_{A}$. Moreover, the
cost of $A$ is at most $3\opt(S)$.\end{lemma}
\begin{proof}
Let $B$ be an optimal BST algorithm on a BST $T_{B}$ with $[n]$
as elements for executing $S$. To construct $T_{A}$, for each element
$i$, we replace $i$ in $T_{B}$ with three elements $i_{L},i_{R},i$
where $i-1<i_{L}<i<i_{R}<i+1$ and $i_{L}$ is always a parent of
$i_{R}$ which is always a parent of $i$. Therefore, $d_{T_{A}}(i)=3d_{T_{B}}(3)$
and $i$ is leaf, for each $i\in[n]$. After each access, if $B$
rearranges $T_{B}$, we rearrange $T_{A}$ accordingly, which costs
at most 3 times as much. 
\end{proof}
\begin{proof}
[Proof of \Cref{thm:interleaving}] Let $P=([a_{1},b_{1}],\dots,[a_{k},b_{k}])$
be the partition of $[n]$ from the theorem. We show a BST algorithm
$A$ on a BST $T_{A}$ with some auxiliary elements such that the
cost for accessing $S$ is at most $\sum_{i=1}^{k}\opt(S_{i})+3\opt(\tilde{S})$. 

To describe $T_{A}$, we construct a BST $\tilde{T}_{A}$ that has
$[k]$ as leaves using Lemma~\ref{lem:integer as leaves}. Then,
for each leaf $i\in[k]$, we replace $i$ with the root of a subtree
$T_{A}^{(i)}$ containing elements in $[a_{i},b_{i}]$. That is, the
auxiliary elements in $\tilde{T}_{A}$ are always above elements in
$T_{A}^{(i)}$ for all $i\in[k]$. To describe the algorithm $A$,
if an element $x\in[a_{i},b_{i}]$ is accessed, we first access the
root of $T_{A}^{(i)}$ using an algorithm from Lemma~{lem:integer as leaves},
and then we access $x$ inside $T_{A}^{(i)}$ using the optimal algorithm
for executing $S_{i}$. By Lemma~{lem:integer as leaves}, the
total cost spent in $\tilde{T}_{A}$ is $3\opt(\tilde{S})$, and the
total cost spent in $T_{A}^{(i)}$ is $\opt(S_{i})$ for each $i$.
\end{proof}

%% file: examples_appendix.tex
\section{Omitted Proofs from Section~\ref{sec:separation}}
\label{sec:app_ex}
\subsection{Proof of Theorem~\ref{thm:tightness}}
All bounds are derived via information-theoretic arguments.
To avoid the need to do probabilistic analysis, we use instead the language of Kolmogorov complexity, which is applicable to a specific input, rather than a distribution on inputs.  
First, we state the following proposition. 

\begin{lemma}[\cite{ChalermsookG0MS15}]
	\label{prop:K less than OPT} For any sequence $S$, let $K(S)$ denote Kolmogorov complexity of $S$. We have 
	$K(S) = O(\opt(S)) $. 
\end{lemma}

We also use the following standard fact in the theory of Kolmogorov complexity. 

\begin{lemma} 
Let $\sset$ be a subset of sequences in $[n]^m$. There exists a sequence $S \in \sset$ with $K(S) = \Omega(\log |\sset|)$.  
\end{lemma}

We are now ready to derive all bounds.

\begin{theorem} [$k$-lazy fingers]  
For each $k$ (possibly a function that depends on $n$), there are infinitely many sequences $X$, for which $\LF^k(S) \cdot \log k\le \opt(S)$.
\end{theorem}
\begin{proof}  
Let $J \subseteq [n]: |J| = k$.  
We choose a sequence $S\in J^m$ that has Kolmogorov complexity at least $K(S) = \Omega(m \log k)$, so we have $\opt(S) = \Omega(m \log k)$.  
Now we argue that the lazy finger cost is low. 
Choose the reference tree as an arbitrary balanced tree $T$.  
Notice that the lazy finger bound is $\LF^k_T(S) =O(m)$: Initial fingers can be chosen to be the location of keys in $J$.   
Afterwards, each key in $J$ is served by its private finger.  
\end{proof}

\begin{theorem}[Monotone]
For each $k$, there are infinitely many sequences $S$ for which $m(S) = k$,   
and $K(S) \geq |S| \cdot \log k$.
\end{theorem}
 
\begin{proof} 
This follows from the result of Regev~\cite{regev1981} (who proved much more general results) which implies that for sufficiently large $n$, the number of permutations $S \in [n]^n$ that avoid $(1,\ldots, k)$ is at least $k^{\Omega(n)}$.
Therefore, there exists a permutation $S$ with $K(S) = \Omega(n \log k)$.  
\end{proof}  

\paragraph*{Pattern avoidance.}
We argue that there is a permutation sequence $S$ of size $n$ that avoids a pattern of size $k$ (an indication of easiness) but nevertheless has high Kolmogorov complexity $K(S)=\Omega(n\sqrt{k})$.

Let $\pi$ be a permutation of size $k$. Let $S_{\pi}(n)$ be the number
of permutations of size $n$ avoiding $\pi$, and let $L(\pi)=\lim_{n\rightarrow\infty}S_{\pi}(n)^{1/n}$.
Let $ex(n,\pi)$ be the maximum number of non-zeros of an $n\times n$
matrix that avoids $\pi$, and let $c(\pi)=ex(n,\pi)/n$. Cibulka
\cite{Cibulka09} shows that $c(\pi)=O(L(\pi)^{4.5})$ and $L(\pi)=c(\pi)^{2}$; for a simpler proof of the second result see~\cite{Fox13}.
\begin{theorem}
	[\cite{Cibulka09,Fox13}]\label{thm:mass vs number}$c(\pi)=L(\pi)^{\Theta(1)}$.
\end{theorem}
In \cite{Fox13}, Fox also shows the surprising
result that there exists a permutation $\pi$ of size $k$ such that $c(\pi)=2^{\Omega(k^{1/4})}$%
It is also shown in~\cite{Fox13} that for any $\pi$
	of size $k$, $c(\pi)=2^{O(k)}$, improving the celebrated result
	of $2^{O(k\log k)}$ by Marcus and Tardos~\cite{MarcusT04}.
Geneson and Tian~\cite{GenesonT15a} improve the lower bound as follows.
\begin{theorem} [\cite{GenesonT15a}]
	\label{thm:lower bound mass}There exists $\pi$ of size $k$ where
	$c(\pi)=2^{\Omega(k^{1/2})}$. 
\end{theorem}

Combining \Cref{thm:mass vs number}
	and \Cref{thm:lower bound mass}, we have that there exists $\pi$
	such that $L(\pi)=2^{\Omega(k^{1/2})}$. By the definition of $L(\pi)$,
	we conclude that there is a permutation $\pi$ of size $k$ and
	infinitely many integers $n$ such that 
$S_{\pi}(n) = \Omega(2^{n\sqrt{k}})$. Let $P$ be such a set
	of permutations. We know that there must be a permutation $S\in P$
	such that the Kolmogorov complexity of $S$ is $K(S)=\Omega(n\sqrt{k})$.

\subsection{Proof of Theorem~\ref{thm:hierarchy}}
Let $n$ be an integer multiple of $k$ and $\ell = n/k$.  
Consider the tilted $k$-by-$\ell$ grid $S_k$. 
The access sequence is $1$, $\ell + 1$, \ldots, $\ell \cdot(k-1) + 1$, $2$, $\ell+2$, \ldots, $(k-1)\ell +2$,\ldots, $(k-1)\ell +\ell$.
It is clear that this sequence avoids $(k+1, k, \ldots, 1)$, so the third part of the theorem follows easily.

We now show that $\LF^k(S_k) = O(n)$.  We write $S_k = (s_1,\dots,s_n)$.
The idea is to partition the keys into blocks, and use each finger to serve only the keys inside blocks.  
In particular, for each $i=1,\ldots, k$, denote by $\bset_i \subseteq [n]$ the set of keys in $[\ell(i-1) +1, \ell \cdot i]$. 
We create a reference tree $T$ and argue that $\LF^k_T(S_k) = O(n)$.
Let $T_0$ be a BST of height $O(\log k)$ and with $k$ leaves. 
Each leaf of $T_0$ corresponds to the keys $\set{\ell \cdot (i-1) +\frac 1 2}_{i=1}^k$. 
The non-leafs of $T_0$ are assigned arbitrary fractional keys that are consistent with the BST properties.
For each $i$, path $P_i$ is defined as a BST with key $\ell \cdot (i-1) +1$ at the root, where for each $j = 0, \ldots, (\ell-1)$, the key $\ell (i-1)+j$ has $\ell (i-1)+ (j+1)$ as its only (right) child.  
The final tree $T$ is obtained by hanging each path $P_i$ as a left subtree of a leaf $\ell \cdot (i-1) +\frac 1 2$.  
The $k$-server strategy is simple: The $i^{th}$ finger only takes care of the elements in block $\bset_i$. 
The cost for the first access in block $\bset_i$ is $O(\log k)$, and afterwards, the cost is only $O(1)$ per access. 
So the total access cost is $O(\frac{n}{k} \log k + n) = O(n)$. 

To see that $\opt(S_k) = O(n)$, let $S_k^{(i)}$ be obtained from $S_k$ by restriction to $\bset_i$ for each $i=1,\dots,k$. Let $\tilde{S}_k=(\tilde{s}_1,\dots,\tilde{s}_n ) \in [k]^n$ where $\tilde{s}_j = i$ iff $s_j \in \bset_i$. 
By \Cref{thm:interleaving}, $\opt(S_k) \le \sum_{i=1}^k \opt(S_k^{(i)}) + O(\opt(\tilde{S}_k))$. 
But, for each $i$, $\opt(S_k^{(i)}) = O(\ell)$ because $S_k^{(i)}$ is just a sequential access of size $\ell$.  $\opt(\tilde{S}_k) = O(k\cdot \ell) = O(n)$ because $\tilde{S}_k$ a sequential access of size $k$ repeated $\ell$ times. Therefore, $\opt(S_k) = O(n)$.

The rest of this section is devoted to proving the following: 
 
\begin{theorem} 
$\LF^{k-1}(S_k) = \Omega(\frac{n}{k} \log (n/k))$
\end{theorem}

A {\em finger configuration} $\vec{f} = (f(1),\ldots, f(k-1)) \in [n]^{k-1}$ specifies to which keys the fingers are currently pointing. 
Let $T$ be a reference tree. 
Any finger strategy can be described by a sequence $\vec{f}_1,\ldots, \vec{f}_n$, where $\vec{f}_t$ is a configuration after element $s_t$ is accessed.  
Just like in a general $k$-server problem, we may assume w.l.o.g.\ the following: 
\begin{fact}
\label{prop:k-server}  
For each time $t$, the configurations $\vec{f}_t$ and $\vec{f}_{t+1}$ differ at exactly one position. In other words, we only move the finger that is used to access $s_{t+1}$.  
\end{fact} 

We see the input sequence $S_k$ as having $\ell$ phases: The first phase contains the subsequence $1, \ell+1,\ldots, \ell(k-1)+1$, and so on. Each phase is a subsequence of length $k$.  

\begin{lemma} 
For each phase $p \in \{1,\ldots, \ell \}$, there is time $t \in [(p-1)\ell+1, p\cdot \ell ]$ such that $s_t$ is accessed by finger $j$ such that $f_{t-1}(j) $ and $f_t(j)$ are in different blocks, and $f_{t-1}(j) < f_t(j)$.  
That is, this finger moves to the block $\bset_{b}$, $b = t \mod k$, from some block $\bset_{b'}$, where $b' < b$, in order to serve $s_t$.  
\end{lemma} 

\begin{proof} 
Suppose not. 
By \Cref{prop:k-server}, this implies that each finger is used to access each element only once in this phase (because accesses in this phase are done in blocks $\bset_1,\ldots, \bset_k$ in this order). 
This is impossible because we only have $k-1$ fingers. 
\end{proof} 

For each phase $p \in [\ell]$, let $t_p$ denote the time for which such a finger moves across the blocks from left to right; if they move more than once, we choose $t_p$ arbitrarily.  
Let $J = \{t_p\}_{p=1}^{\ell}$.   
For each finger $j \in [k-1]$, each block $i \in [k]$ and block $i' \in [k]: i < i'$, 
let $J(j,i,i')$ be the set containing the time $t$ for which finger $f(j)$ is moved from block $\bset_i$ to block $\bset_{i'}$ to access $s_t$. Let $c(j,i,i') = |J(j,i,i')|$. 
Notice that $\sum_{j,i,i'} c(j,i,i') = \frac{n}{k} = \ell$, due to the lemma. 
Let $P(j,i,i')$ denote the phases $p$ for which $t_p \in J(j,i,i')$.

\begin{lemma} 
$\sum_{j, i, i': c(j,i,i') \geq 16} c(j,i,i') \geq n/2k$ if $n = \Omega(k^4)$.  
\end{lemma} 
\begin{proof}
There are only at most $k^3$ triples $(j,i,i')$, so the terms for which $c(j,i,i') < 16$ contribute to the sum at most $16k^3$. 
This means that the sum of the remaining is at least $n/k - 16k^3 \geq n/2k$ if $n$ satisfies $n= \Omega(k^4)$.  
\end{proof}

From now on, we consider the sets $J'$ and $J'(j,i,i')$ that only concern those $c(j,i,i')$ with $c(j,i,i') \geq 16$ instead.  

\begin{lemma}
\label{lem:key-lemma}
There is a constant $\eta >0$ such that the total access cost during the phases $P(j,i,i')$ is at least $\eta c(j,i,i') \log c(j,i,i')$.  
\end{lemma} 

Once we have this lemma, everything is done. 
Since the function $g(x) = x \log x$ is convex, we apply Jensen's inequality to obtain: 
\[\frac{1}{|J'|} \sum_{j,i,i'} \eta c(j,i,i')\log c(j,i,i') \geq \eta (\frac{n}{2k|J'|}) \log (n/2k|J'|). \]
Note that the left side is the term ${\mathbb E}[g(x)]$, while the right side is $g({\mathbb E}(x))$. Therefore, the total access cost is at least $\frac{\eta n}{8k} \log (n/2k)$, due to the fact that $|J'| \leq k^3$.
We now prove the lemma. 

\begin{proof}[Proof of Lemma~\ref{lem:key-lemma}] 
We recall that, in the phases $P(j,i,i')$, the finger-$j$ moves from block $\bset_i$ to $\bset_{i'}$ to serve the request at corresponding time.  
For simplicity of notation, we use $\tilde J$ and $C$ to denote $J(j,i,i')$ and $c(j,i,i')$ respectively. 
Also, we use $\tilde f$ to denote the finger-$j$.  
For each $t \in \tilde J$, let $a_t \in \bset_i$ be the key for which the finger $\tilde f$ moves from $a_t$ to $s_t$ when accessing $s_t \in \bset_{i'}$. 
Let $\tilde J = \{t_1,\ldots, t_{C}\}$ such that $a_{t_1} < a_{t_2} <  \ldots < a_{t_{C}}$. 
Let $R$ be the lowest common ancestor in $T$ of keys in $[a_{t_{\lfloor C/2 \rfloor } +1}, a_{t_{C}}]$.

\begin{lemma} 
\label{lem:lastclaim}
For each $r \in \{1,\ldots, \lfloor C/2 \rfloor \}$, the access cost of $s_{t_r}$  and $s_{t_{C -r}}$ is together at least 
$\min \{d_T({R}, s_{t_r}), d_T(R, s_{t_{C -r}})\}$. 
\end{lemma} 

\begin{proof} 
 
Let $u_r$ be the lowest common ancestor between $a_{t_r}$ and $s_{t_r}$. If $s_{t_r}$ is in the subtree rooted at $R$, then the cost must also be at least $d_T(u_r, s_{t_r}) \geq d_T(R,s_{t_r})$ as $u_r$ must be an ancestor of $R$ (because $a_{t_r} < a_{t_{\lfloor C/2 \rfloor }} < a_{t_{C}} < s_{t_r}$). Otherwise, we know that $s_{t_r}$ is outside of the subtree rooted at $R$, and so is $s_{t_{C- r}}$. On the other hand, $a_{t_{C-r}}$ is in such subtree, so moving the finger from $a_{t_{C -r}}$ to $s_{t_{C - r}}$ must touch ${R}$, therefore costing at least $d_T({R}, s_{t_{C-r}})$. 

\end{proof} 

Lemma~\ref{lem:lastclaim} implies that, for each $r=1,\ldots, \lfloor C/2 \rfloor $, we pay the distance between some element $v_r \in \set{s_{t_r}, s_{t_{C-r}}}$ to $R$. 
The total such costs would be $\sum_{r} d_T(R, v_r)$. 
Applying the fact that (i) $v_r$'s are different and (ii) there are at most $3^d$ vertices at distance $d$ from a vertex $R$, we conclude that this sum is at least $\sum_r d_T(R, v_r) \geq \Omega(C \log C)$.   
\end{proof}

\input{WSandLF}

%% file: WSandLF.tex
{\renewcommand {\tset}{T}

\subsection{Working set and $k$-lazy finger bounds are incomparable}

We show the following.

\begin{theorem}\ \\
\begin{compactitem}
\item There exists a sequence $S$ such that $\WS(S) = o(\LF^{k}(S))$, and
\item There exists a sequence $S^{'}$ such that $\LF^{k}(S^{'}) = o(\WS(S^{'}))$.
\end{compactitem}
\end{theorem}

The sequence $S^{'}$ above is straightforward: For $k=1$, just consider the sequential access $1,\dots,n$ repeated $m/n$ times. For $m$ large enough, the working set bound is $\Omega(m \log n)$. However, if we start with the finger on the root of the tree which is just a path, then the lazy finger bound is $O(m)$. The $k$-lazy finger bound is always less than lazy finger bound, so this sequence works for the second part of the theorem.

The existence of the sequence $S$ is slightly more involved (the special case for $k=1$ was proved in~\cite{BoseDIL14}), and is guaranteed by the following theorem, the proof of which comprises the remainder of this section.

\begin{theorem}\label{wslessklf}
For all $k = O(n^{1/2 - \eps})$, there exists a sequence $S$ of length $m$ such that $\WS(S) = O(m \log k)$ whereas $\LF^{k}(S) = \Omega(m \log (n/k))$. 
\end{theorem}

We construct a random sequence $S$ and show that while $\WS(S) = O(m \log k)$ with probability one, the probability that there exists a tree $\tset$ such that $\LF^{k}_{T}(S) \leq c m\log_{3} (n/k)$ is less than $1/2$ for some constant $c<1$. This implies the existence of a sequence $S$ such that for all trees $\tset$, $\LF^{k}_{T}(S) = \Omega(m \log (n/k))$.

The sequence is as follows. We have $Y$ phases. In each phase we select $2k$ elements $R_{i} = \{r^{i}_{j}\}_{j=1}^{2k}$ uniformly at random from $[n]$. We order them arbitrarily in a sequence $S_{i}$, and access $[S_{i}]^{X/2k}$ (access $S_{i}$ $X/2k$ times). The final sequence $S$ is a concatenation of the sequences $[S_{i}]^{X/2k}$ for $1 \leq i \leq Y$. Each phase has $X$ accesses, for a total of $m= XY$ accesses overall. We will choose $X$ and $Y$ appropriately later.

\noindent{\bf Working set bound.} One easily observes that $\WS(S) =O(Y(2k \log n + (X-2k) \log (2k)))$, because after the first $2k$ accesses in a phase, the working set is always of size $2k$. We choose $X$ such that the second term dominates the first, say $X \geq 5 k \frac{\log n}{\log 2k}$. We then have that the working set bound is $O(XY \log k) = O(m \log k)$, with probability one.

\noindent{\bf $k$-lazy finger bound.} Fix a BST $\tset$. We classify the selection of the set $R_{i}$ as being $d$-good for $\tset$ if there exists a pair $r^{i}_{j}, r^{i}_{\ell} \in R_{i}$ such that their distance in $\tset$ is less than $d$. The following lemma bounds the probability of a random selection being $d$-good for $\tset$. 

\begin{lemma} Let $\tset$ be any BST. 
The probability that $R_{i}$ is $d$-good for $\tset$ is at most $8k^{2}3^{d}/n$.
\end{lemma}

\begin{proof} We may assume $8 k^2 3^d/n < 1$ as the claim is void otherwise. We compute the probability that a selection $R_{i}$ is not $d$-good first. This happens if and only if the balls of radius $d$ around every element $r^{i}_{j}$ are disjoint. The volume of such a ball is at most $3^d$, so we can bound this probability as 

\begin{eqnarray}
P[R_{i}\text{ is not d-good for }\tset] &=& \displaystyle\Pi_{i=1}^{2k-1} \left(1-\frac{i3^d}{n}\right) \notag \\
								&\geq& \left( 1-\frac{2k3^d}{n}\right)^{2k} \notag \\
\Rightarrow P[R_{i}\text{ is d-good for }\tset] &\leq& 1- \left( 1-\frac{2k3^d}{n}\right)^{2k} \notag \\
									&=& 1- \exp\left(2k \ln \left( 1-\frac{2k3^d}{n}\right)\right) \notag \\
									&\leq& 1 - \exp\left(-8k^{2}3^{d}/n\right) \notag \\
									&\leq& 8k^{2}3^{d}/n, \notag
\end{eqnarray}
where the last two inequalities follow from $\ln (1-x) > -2x$ for $x \le 1/2$ (note that $8 k^2 3^d/n < 1$ implies 
$2k 3^d/n \le 1/2$) and $e^x > 1+x$, respectively.
\end{proof}

Observe that if $R_{i}$ is not $d$-good, then the $k$-lazy finger bound of the access sequence $[S_{i}]^{X/2k}$ is $\Omega(d(X-k))=\Omega(dX)$. This is because in every occurrence of $S_{i}$, there will be some $k$ elements out of the $2k$ total that will be outside the $d$-radius balls centered at the current $k$ fingers.

We call the entire sequence $S$ $d$-good for $\tset$ if at least half of the sets $R_{i}$ are $d$-good for $\tset$. Thus if $S$ is not $d$-good, then $\LF^{k}_{\tset}(S) = \Omega(XYd)$.

\begin{lemma}
$P[S\text{ is d-good for }\tset] \leq \left(\frac{32k^{2}3^{d}}{n} \right)^{Y/2}$. 
\end{lemma}

\begin{proof}
By the previous lemma and by definition of goodness of $S$, we have that

\begin{eqnarray}
P[S\text{ is d-good for }\tset] &\leq& {Y \choose Y/2}\left(\frac{8k^{2}3^{d}}{n} \right)^{Y/2} \notag \\
			& \leq& 4^{Y/2}\left(\frac{8k^{2}3^{d}}{n} \right)^{Y/2} \notag \\
			&=& \left(\frac{32k^{2}3^{d}}{n} \right)^{Y/2}.\notag
\end{eqnarray}
\end{proof}

The theorem now follows easily. Taking a union bound over all BSTs on $[n]$, we have 
\[ P[S\text{ is d-good for some BST }\tset] \leq 4^{n}\left(\frac{32k^{2}3^{d}}{n} \right)^{Y/2}.
\]

Now set $Y=2n$. We have that 
\[ P[\exists \text{ a BST }\tset:\LF^{k}_{\tset}(S) \leq md/4] \leq 4^{n}\left(\frac{32k^{2}3^{d}}{n} \right)^{n}.
\]

Putting $d = \log_{3} \frac{n}{256k^{2}}$ gives that for some constant $c<1$,
\begin{eqnarray}
 P[\exists \text{ a BST }\tset:\LF^{k}_{\tset}(S) \leq c(m \log (n/k))] &\leq& 4^{n}\left(\frac{32k^{2}3^{d}}{n} \right)^{n} = 1/2 \notag
\end{eqnarray}
which implies that with probability at least $1/2$ one of the sequences in our random construction will have $k$-lazy finger bound that is $\Omega(m \log (n/k))$. The working set bound is always $O(m \log k)$. This establishes the theorem.

}